\documentclass[11pt]{article} 
\usepackage{mathtools,amssymb,amsmath,amsthm,graphicx,natbib}

\newcommand{\btheta}{\boldsymbol{\theta}}
\DeclareMathOperator*{\argmaxA}{arg\,max}

\usepackage[margin=1in]{geometry}
\usepackage{setspace}
\theoremstyle{theorem}
\newtheorem{theorem}{Theorem}

\theoremstyle{prop}
\newtheorem{prop}{Proposition}

\theoremstyle{lem}
\newtheorem{lem}{Lemma}

\usepackage{caption}

\begin{document}

\title{Combining Biomarkers by Maximizing the True Positive Rate for a Fixed False Positive Rate}
\date{\vspace{-10ex}}

\maketitle

\begin{center}
{\normalsize
ALLISON MEISNER\\[2pt]
\textit{Department of Biostatistics, Johns Hopkins University, Baltimore, MD, USA}\\[2pt]
{ameisne1@jhu.edu}\\[4pt]
MARCO CARONE\\[2pt]
\textit{Department of Biostatistics, University of Washington, Seattle, WA, USA} \\[4pt]
MARGARET S. PEPE\\[2pt]
\textit{Public Health Sciences Division, Fred Hutchinson Cancer Research Center, Seattle, WA, USA} \\[4pt]
KATHLEEN F. KERR \\[2pt]
\textit{Department of Biostatistics, University of Washington, Seattle, WA, USA}}
\end{center}

\vspace{5ex}

\begin{abstract}
{Biomarkers abound in many areas of clinical research, and often investigators are interested in combining them for diagnosis, prognosis, or screening. In many applications, the true positive rate for a biomarker combination at a prespecified, clinically acceptable false positive rate is the most relevant measure of predictive capacity. We propose a distribution-free method for constructing biomarker combinations by maximizing the true positive rate while constraining the false positive rate. Theoretical results demonstrate desirable properties of biomarker combinations produced by the new method. In simulations, the biomarker combination provided by our method demonstrated improved operating characteristics in a variety of scenarios when compared with alternative methods for constructing biomarker combinations.}
{Biomarker combinations; False positive rate; Sensitivity; Specificity; True positive rate.}
\end{abstract}

\section{Introduction}

As the number of available biomarkers has grown, so has the interest in combining them for the purposes of diagnosis, prognosis, or screening. Over the past decade, much work has been done to develop methods for constructing biomarker combinations by targeting measures of performance, including those related to the receiver operating characteristic, or ROC, curve. This is in contrast to more traditional methods that construct biomarker combinations by optimizing general global fit criteria, such as the maximum likelihood approach. While methods to construct both linear and nonlinear combinations have been proposed, linear biomarker combinations are more common  than nonlinear combinations, due to their greater interpretability and ease of construction~\citep{wang2011,hsu2013}.

Although the area under the ROC curve, the AUC, is arguably the most popular way to summarize the ROC curve, there is often interest in identifying a biomarker combination with a high true positive rate (TPR), the proportion of correctly classified diseased individuals, while setting the false positive rate (FPR), the proportion of incorrectly classified nondiseased individuals, at some clinically acceptable level. A common practice among applied researchers is to construct linear biomarker combinations using logistic regression, and then calculate the TPR for the prespecified FPR, e.g.,~\cite{moore2008}. While methods for constructing biomarker combinations by maximizing the AUC or the partial AUC have been developed, these methods do not directly target the TPR for a specified FPR. 

We propose a distribution-free method for constructing linear biomarker combinations by maximizing the TPR while constraining the FPR. We demonstrate desirable theoretical properties of the resulting combination, and provide empirical evidence of good small-sample performance through simulations. To illustrate our method, we consider data from a prospective study of diabetes mellitus in 532 adult women with Pima Indian heritage~\citep{smithpima}. Several variables were measured for each participant, and criteria from the World Health Organization were used to identify women who developed diabetes. A primary goal of the study was to predict the onset of diabetes within five years. 

\section{Background}

\subsection{ROC Curve and Related Measures}

The ROC curve provides a means to evaluate the ability of a biomarker or, equivalently, biomarker combination $Z$ to identify individuals who have or will experience a binary outcome $D$. For example, in a diagnostic setting, $D$ denotes the presence or absence of disease and $Z$ may be used to identify individuals with the disease. The ROC curve provides information about how well the biomarker discriminates between individuals who have or will experience the outcome, that is, the cases, and individuals who do not have or will not experience the outcome, that is, the controls~\citep{pepebook}. Mathematically, if larger values of $Z$ are more indicative of having or experiencing the outcome, for each threshold $\delta$ we can define the TPR as $P(Z > \delta\mid D=1)$ and the FPR as $P(Z > \delta\mid D=0)$~\citep{pepebook}. For a given $\delta$, the TPR is also referred to as the sensitivity, and one minus the specificity equals the FPR~\citep{pepebook}. The ROC curve is a plot of the TPR versus the FPR as $\delta$ ranges over all possible values; as such, it is non-decreasing and takes values in the unit square~\citep{pepebook}. A perfect biomarker has an ROC curve that reaches the upper left corner of the unit square, and a useless biomarker has an ROC curve on the 45-degree line~\citep{pepebook}. 

The most common summary of the ROC curve is the AUC, the area under the ROC curve. The AUC ranges between 0.5 for a useless biomarker and 1 for a perfect biomarker~\citep{pepebook}. The AUC has a probabilistic interpretation: it is the probability that the biomarker value for a randomly chosen case is larger than that for a randomly chosen control, assuming that higher biomarker values are more indicative of having or experiencing the outcome~\citep{pepebook}. Both the ROC curve and the AUC are invariant to monotone increasing transformations of the biomarker $Z$~\citep{pepebook}.

The AUC summarizes the entire ROC curve, but in many situations it is more appropriate to only consider certain FPR values. For example, screening tests require a very low FPR, while diagnostic tests for fatal diseases may allow for a slightly higher FPR if the corresponding TPR is very high~\citep{hsu2013}. Such considerations led to the development of the partial AUC, the area under the ROC curve over some range $(t_0, t_1)$ of FPR values~\citep{pepebook}. Rather than considering a range of FPR values, there may be interest in fixing the FPR at a single value, determining the corresponding threshold $\delta$, and evaluating the TPR for that threshold. As opposed to the AUC and the partial AUC, this method returns a single classifier, or decision rule, which may appeal to researchers seeking a tool for clinical decision-making. 

\subsection{Biomarker Combinations} 

Many methods to combine biomarkers have been proposed, and they can be divided into two categories. The first includes indirect methods that seek to optimize a measure other than the performance measure of interest, while the second category includes direct methods that optimize the target performance measure. We focus on the latter.

Targeting the entire ROC curve (that is, constructing a combination that produces an ROC curve that dominates the ROC curve for all other linear combinations at all points) is very challenging and is only possible under special circumstances.~\cite{su1993} demonstrated that, when the vector $\textbf{X}$ of biomarkers has a multivariate normal distribution conditional on $D$ with proportional covariance matrices, it is possible to identify the linear combination that maximizes the TPR uniformly over the entire range of FPRs; this linear combination is Fisher's linear discriminant function.

~\cite{mcintosh2002} used the Neyman-Pearson lemma to demonstrate optimality (in terms of the ROC curve) of the likelihood ratio function and, consequently, of the risk score $P(D=1|\textbf{X}=\textbf{x})$ and monotone transformations of $P(D=1|\textbf{X}=\textbf{x})$. Thus, if the biomarkers are conditionally multivariate normal and the $D$-specific covariance matrices are equal, the optimal linear combination dominates not just every other linear combination, but also every nonlinear combination. This results from the fact that in this case, the linear logistic model $\mbox{logit}\lbrace P(D=1|\textbf{X}=\textbf{x})\rbrace = {\btheta}^{\top}\textbf{x}$ holds for some $p$-dimensional ${\btheta}$, where $p$ is the dimension of $\textbf{X}$. If the covariance matrices are proportional but not equal, the likelihood ratio is a nonlinear function of the biomarkers, as shown in the Appendix A for $p=2$, and the optimal biomarker combination with respect to the ROC curve is nonlinear. 

In general, there is no linear combination that dominates all others in terms of the TPR over the entire range of FPR values~\citep{su1993,anderson1962}. Thus, methods to optimize the AUC have been proposed. When the biomarkers are conditionally multivariate normal with nonproportional covariance matrices,~\cite{su1993} gave an explicit form for the best linear combination with respect to the AUC. Others have targeted the AUC without any assumption on the distribution of the biomarkers; many of these methods rely on smooth approximations to the empirical AUC, which involves indicator functions~\citep*{ma2007, fong2016, lin2011}.   

Acknowledging that often only a range of FPR values is of interest clinically, methods have been proposed to target the partial AUC for some FPR range $(t_0, t_1)$. Some methods make parametric assumptions about the joint distribution of the biomarkers~\citep{yu2015,hsu2013,yan2018} while others do not~\citep{wang2011,komori2010,yan2018}. The latter group of methods generally uses a smooth approximation to the partial AUC, similar to some of the methods that aim to maximize the AUC~\citep{wang2011,komori2010,yan2018}. One challenge faced in partial AUC maximization is that for narrow intervals, that is, when $t_0$ is close to $t_1$, the partial AUC is often very close to 0, which can make optimization difficult~\citep{hsu2013}. 


Some work in constructing biomarker combinations by maximizing the TPR has been done for conditionally multivariate normal biomarkers. In this setting, procedures for constructing a linear combination that maximizes the TPR for a fixed FPR~\citep{anderson1962,gao2008} as well as methods for constructing a linear combination by maximizing the TPR for a range of FPR values~\citep*{liu2005b} have been proposed. Importantly, in the method proposed by~\cite{liu2005b}, the range of FPR values over which the fitted combination is optimal may depend on the combination itself; that is, the range of FPR values may be determined by the combination and so may not be fixed in advance. Thus, this method does not optimize the TPR for a prespecified FPR.~\cite{baker2000} proposed a flexible nonparametric method for combining biomarkers by optimizing the ROC curve over a narrow target region of FPR values. However, this method is not well-suited to situations in which more than a few biomarkers are to be combined.  

An important benefit of constructing linear biomarker combinations by targeting the performance measure of interest is that the performance of the combination will be at least as good as the performance of the individual biomarkers~\citep*{pepe2006}. Indeed, several authors have recommended matching the objective function to the performance measure, i.e., constructing biomarker combinations by optimizing the relevant measure of performance~\citep{hwang2013,liu2005b,wang2011,ricamato2011}. To that end, we propose a distribution-free method to construct biomarker combinations by maximizing the TPR for a given FPR. 

Figure~\ref{fig1} illustrates the importance of targeting the measure of interest in constructing biomarker combinations. In this example, combinations of three biomarkers are constructed by (i) maximizing the logistic likelihood, (ii) maximizing the AUC via the \texttt{optAUC} package in \texttt{R} (i.e., the method of~\cite{huang2011}), and (iii) maximizing the TPR for an FPR of 20\% using the proposed method. The ROC curves for the three combinations differ markedly near the prespecified FPR of 20\%. In particular, the TPRs at an FPR of 20\% for the three combinations are 18.0\%, 24.0\%, and 34.0\% for maximum likelihood, AUC optimization, and maximization of the TPR for a given FPR, respectively. This example highlights the utility of methods that target the TPR for a specific FPR as opposed to methods that target other measures.

\begin{figure}[ht]
\begin{center}
\includegraphics[width=3.8in]{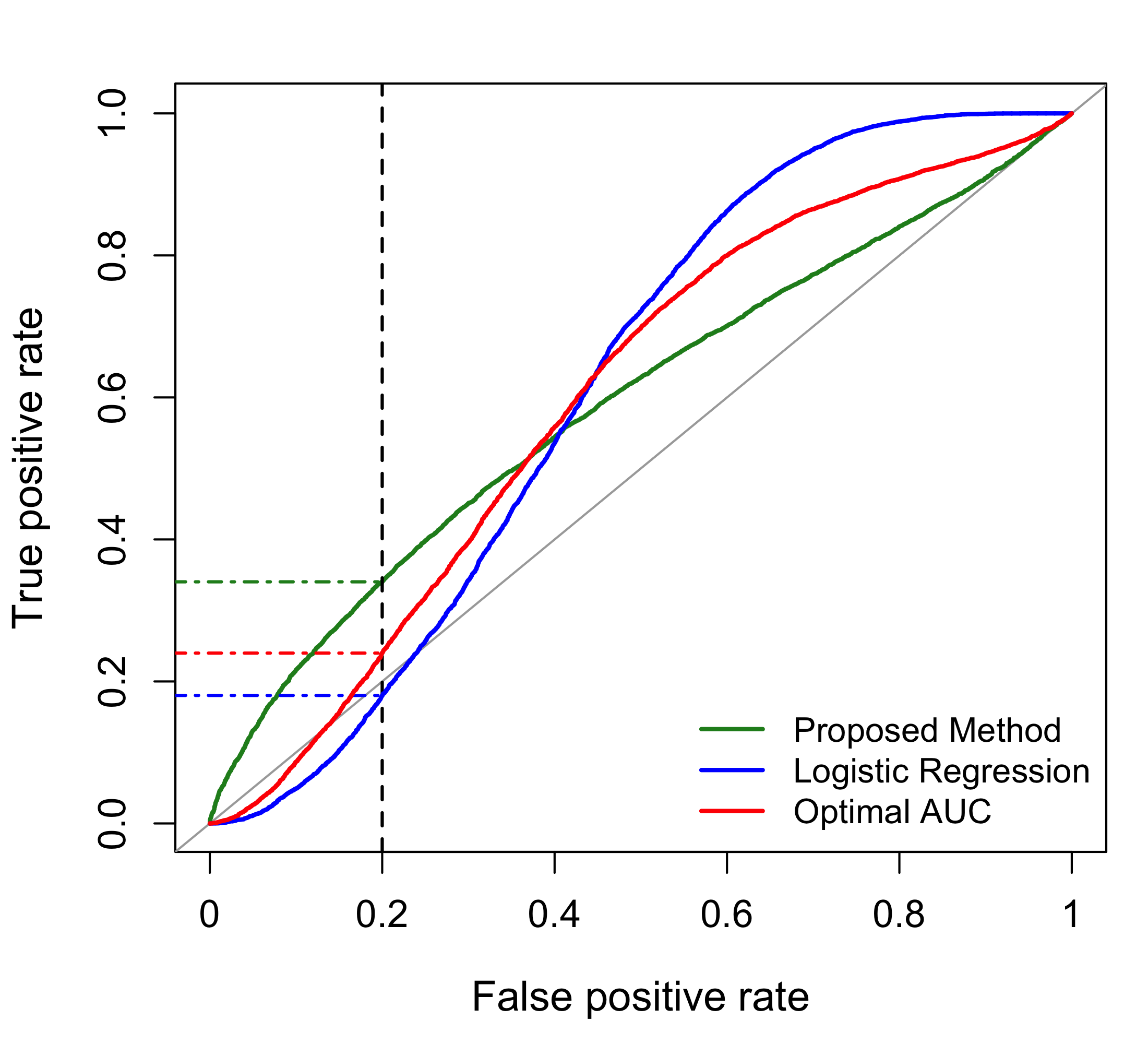}
\caption{Biomarker combinations obtained by targeting different measures. In this illustrative example, there is interest in the TPR for an FPR of 20\%. Combinations of three biomarkers were constructed in training data (400 cases and 400 controls) and evaluated in a large test dataset (10,000 cases and 10,000 controls). Two of the biomarkers, $X_1$ and $X_2$, were distributed as conditional bivariate lognormal random variables such that, among controls, $E(\mbox{log}(X_1)) = 1.1, E(\mbox{log}(X_2)) = 1.1, Var(\mbox{log}(X_1)) = 0.04, Var(\mbox{log}(X_2)) = 0.5, Cov(\mbox{log}(X_1),\mbox{log}(X_2)) = 0.09$ and, among cases, $E(\mbox{log}(X_1)) = 1, E(\mbox{log}(X_2)) = 1, Var(\mbox{log}(X_1)) = 0.05, Var(\mbox{log}(X_2)) = 0.05, Cov(\mbox{log}(X_1),\mbox{log}(X_2)) = 0.015$. The third biomarker, $X_3$, was distributed as an independent lognormal random variable such that, among both cases and controls, $E(\mbox{log}(X_3)) = 1.65, Var(\mbox{log}(X_3)) = 4.66$. The combinations were constructed by maximizing the TPR for an FPR of 20\% (``Optimal TPR''; green), by maximizing the logistic likelihood (``Logistic Regression''; blue), and by maximizing the AUC (``Optimal AUC''; red). The gray line indicates the ROC curve for a useless marker (FPR and TPR are equal), the black dashed line denotes an FPR of 20\%, and the dot-dashed lines indicate the TPRs for an FPR of 20\%. This figure appears in color in the electronic version of this article.}
\label{fig1}
\end{center}
\end{figure}

\section{Methodology}

\subsection{Description}

Cases are denoted by the subscript $1$ and controls are denoted by the subscript $0$. Let $\textbf{X}_{1i}$ denote the vector of biomarkers for the $i^{th}$ case, and let $\textbf{X}_{0j}$ denote the vector of biomarkers for the $j^{th}$ control. 

We propose constructing a linear biomarker combination of the form ${\btheta}^{\top}\textbf{X}$ for a $p$-dimensional $\textbf{X}$ by maximizing the TPR when the FPR is below some prespecified, clinically acceptable value $t$. We define the true and false positive rates for a given $\textbf{X}$ as a function of ${\btheta}$ and $\delta$:
\[ \mbox{TPR}({\btheta}, \delta) = P({\btheta}^{\top}\textbf{X} > \delta|D=1),\;\;\;\; \mbox{FPR}({\btheta}, \delta) = P({\btheta}^{\top}\textbf{X} > \delta|D=0).\]
Since the true and false positive rates for a given combination ${\btheta}$ and threshold $\delta$ are invariant to scaling of the parameters $({\btheta}, \delta)$, we must restrict $({\btheta}, \delta)$ to ensure identifiability. Specifically, we constrain $||\boldsymbol{\theta}|| = 1$ as in \cite{fong2016}. For any fixed $t \in (0,1)$, we can consider
\begin{equation*}
(\btheta_t, \delta_t) \in \argmaxA_{({\btheta}, \delta) \in \Omega_t} \mbox{TPR}({\btheta}, \delta),
\end{equation*}
where $\Omega_t = \lbrace {\btheta} \in \mathbb{R}^p, \delta \in \mathbb{R}: ||{\btheta}||=1, \mbox{FPR}({\btheta},\delta) \leq t \rbrace.$ This provides the optimal combination $\btheta_t$ and threshold $\delta_t$. We define $(\btheta_t, \delta_t)$ to be an element of $\argmaxA_{({\btheta}, \delta) \in \Omega_t} \mbox{TPR}({\btheta}, \delta)$, where $\argmaxA_{({\btheta}, \delta) \in \Omega_t} \mbox{TPR}({\btheta}, \delta)$ may be a set. 

Of course, in practice, the true and false positive rates are unknown, so $\btheta_t$ and $\delta_t$ cannot be computed. We can replace these unknowns by their empirical estimates,
\begin{equation*}
\hat{\mbox{TPR}}_{n_1}({\btheta}, \delta) = \frac{1}{n_1} \sum_{i=1}^{n_1} 1({\btheta}^{\top}\textbf{X}_{1i} > \delta),\;\; \hat{\mbox{FPR}}_{n_0}({\btheta}, \delta) = \frac{1}{n_0} \sum_{j=1}^{n_0} 1({\btheta}^{\top}\textbf{X}_{0j} > \delta),
\end{equation*}
where $n_1$ is the number of cases and $n_0$ is the number of controls, giving the total sample size $n = n_1 + n_0$. We can then define
\begin{equation*}
(\hat{{\btheta}}_{t}, \hat{\delta}_{t}) \in \argmaxA_{(\btheta, \delta) \in \hat{\Omega}_{t,n_0}} \hat{\mbox{TPR}}_{n_1}({\btheta}, \delta) 
\end{equation*}
where $\hat{\Omega}_{t,n_0} = \lbrace {\btheta} \in \mathbb{R}^p, \delta \in \mathbb{R}: ||{\btheta}||=1, \hat{\mbox{FPR}}_{n_0}({\btheta},\delta) \leq t \rbrace.$ It is possible to conduct a grid search over $({\btheta}, \delta)$ to perform this constrained optimization, though this becomes computationally demanding when combining more than two or three biomarkers.  

Since the objective function involves indicator functions, it is not a smooth function of the parameters $({\btheta}, \delta)$ and thus not amenable to derivative-based methods. However, smooth approximations to indicator functions  have been used for AUC maximization~\citep{ma2007, fong2016, lin2011}. One such smooth approximation is $1(w > 0) \approx \Phi(w/h)$, where $\Phi$ is the standard normal distribution function, and $h$ is a tuning parameter representing the trade-off between approximation accuracy and estimation feasibility such that $h$ tends to zero as the sample size grows~\citep{lin2011}. We can use this smooth approximation to implement the method described above, writing the smooth approximations to the empirical true and false positive rates as
\begin{equation*}
\tilde{\mbox{TPR}}_{n_1}({\btheta}, \delta) = \frac{1}{n_1} \sum_{i=1}^{n_1} \Phi\left(\frac{{\btheta}^{\top}\textbf{X}_{1i}-\delta}{h}\right), \;\;\tilde{\mbox{FPR}}_{n_0}({\btheta}, \delta) = \frac{1}{n_0} \sum_{j=1}^{n_0} \Phi\left(\frac{{\btheta}^{\top}\textbf{X}_{0j}-\delta}{h}\right).
\end{equation*}

Thus, we propose to compute
\begin{equation}
(\tilde{\btheta}_t, \tilde{\delta}_t) \in \argmaxA_{({\btheta}, \delta) \in \tilde{\Omega}_{t,n_0}} \tilde{\mbox{TPR}}_{n_1}({\btheta}, \delta), \label{defeqn}
\end{equation}
where $\tilde{\Omega}_{t,n_0} = \lbrace {\btheta} \in \mathbb{R}^p, \delta \in \mathbb{R}: ||{\btheta}|| = 1, \tilde{\mbox{FPR}}_{n_0}({\btheta},\delta) \leq t \rbrace.$ Since both $\tilde{\mbox{TPR}}_{n_1}$ and $ \tilde{\mbox{FPR}}_{n_0}$ are smooth functions, we can use gradient-based methods that incorporate the necessary constraints, e.g., Lagrange multipliers. In particular, $(\tilde{\btheta}_t, \tilde{\delta}_t)$ can be obtained using existing software for constrained optimization of smooth functions, such as the \texttt{Rsolnp} package in \texttt{R}. An \texttt{R} package including code for our method based on \texttt{Rsolnp}, \texttt{maxTPR}, is available on CRAN. Other details related to implementation, including the choice of tuning parameter $h$, are discussed below. 

\subsection{Asymptotic Properties} 

We present a theorem establishing that, under certain conditions, the combination obtained by optimizing the smooth approximation to the empirical TPR while constraining the smooth approximation to the empirical FPR has desirable operating characteristics. In particular, its FPR is bounded almost surely by the acceptable level $t$ in large samples. In addition, its TPR converges almost surely to the supremum of the TPR over the set where the FPR is constrained. We focus on the operating characteristics of $(\tilde{\btheta}_t,\tilde{\delta}_t)$ since these are of primary interest to clinicians.  

Rather than enforcing $(\tilde{\btheta}_t,\tilde{\delta}_t)$ to be a strict maximizer, in the theoretical study below we allow it to be a near-maximizer of $\tilde{\mbox{TPR}}_{n_1}(\btheta,\delta)$ within $\tilde{\Omega}_{t,n_0}$ in the sense that \[\tilde{\mbox{TPR}}_{n_1}(\tilde{{\btheta}}_{t},\tilde{\delta}_{t})\ \geq\ \sup_{({\btheta}, \delta) \in \tilde{\Omega}_{t,n_0}} \tilde{\mbox{TPR}}_{n_1}({\btheta}, \delta) - a_n\ ,\] where $a_n$ is a decreasing sequence of positive real numbers tending to zero. This provides some flexibility to accommodate situations in which a strict maximizer either does not exist or is numerically difficult to identify. 

Before stating our key theorem, we give the following conditions. 

\begin{enumerate}
\item[(1)] Observations are randomly sampled conditional on disease status $D$, and the group sizes tend to infinity proportionally, in the sense that $n = n_1 + n_0 \rightarrow \infty$ and $n_1/n_0 \rightarrow \rho \in (0,1)$. 
\item[(2)] For each $d\in\{0,1\}$, observations $\textbf{X}_{di}$, $i=1,2,\ldots,n_d$, are independent and identically distributed $p$-dimensional random vectors with distribution function $F_d$. 
\item[(3)] For each $d\in\{0,1\}$, no proper linear subspace $S\subset\mathbb{R}^p$ is such that $P(\textbf{X} \in S \mid D=d) = 1.$ 
\item[(4)] For each $d\in\{0,1\}$, the distribution and quantile functions of ${\btheta}^{\top}\textbf{X}$ given $D=d$ are globally Lipschitz continuous uniformly over ${\btheta}\in \mathbb{R}^p$ such that $\|{\btheta}\|=1$. 
\item[(5)] The map $(\btheta, \delta) \mapsto \mbox{TPR}({\btheta}, \delta)$ is globally Lipschitz continuous over $\Omega = \lbrace {\btheta} \in \mathbb{R}^p, \delta \in \mathbb{R}: ||{\btheta}|| = 1 \rbrace$. 
\end{enumerate}

\begin{theorem}
\label{theorem1}
Under conditions (1)--(5), for every fixed $t \in (0,1)$, we have that 
\begin{itemize}
\item[(a)] $\limsup_n\mbox{FPR}(\tilde{{\btheta}}_{t},\tilde{\delta}_{t}) \leq t$ almost surely; and 
\item[(b)] $| \mbox{TPR}(\tilde{{\btheta}}_{t}, \tilde{\delta}_{t}) - \sup_{({\btheta}, \delta) \in \Omega_t} \mbox{TPR}({\btheta}, \delta)|$ tends to zero almost surely.
\end{itemize}
\end{theorem}

The proof of Theorem~\ref{theorem1} is given in Appendix B. The proof relies on two lemmas, also in Appendix B. Lemma 1 demonstrates almost sure convergence to zero of the difference between the supremum of a function over a fixed set and the supremum of the function over a stochastic set that converges to the fixed set in an appropriate sense. Lemma 2 establishes the almost sure uniform convergence to zero of the difference between the FPR and the smooth approximation to the empirical FPR and the difference between the TPR and the smooth approximation to the empirical TPR. The proof of Theorem~\ref{theorem1} demonstrates that Lemma 1 holds for the relevant function and sets, relying in part on the conclusions of Lemma 2. The conclusions of Lemmas 1 and 2 then demonstrate the claims of Theorem~\ref{theorem1}.

\subsection{Implementation Details}

Certain considerations must be addressed to implement the proposed method, including the choice of tuning parameter $h$ and starting values $(\tilde{{\btheta}}, \tilde{\delta})$ for the optimization routine. In using similar methods to maximize the AUC,~\cite{lin2011} proposed using $h = \tilde{\sigma} n^{-1/3}$, where $\tilde{\sigma}$ is the sample standard error of $\tilde{{\btheta}}^{\top}\textbf{X}$. In simulations, we considered both $h = \tilde{\sigma} n^{-1/3}$ and $h = \tilde{\sigma} n^{-1/2}$ and found similar behavior for the convergence of the optimization routine. Thus, we use $h = \tilde{\sigma} n^{-1/2}$. We must also identify initial values $(\tilde{{\btheta}}, \tilde{\delta})$ for our procedure. As done in~\cite{fong2016}, we use normalized estimates from robust logistic regression, which is described in greater detail below. Based on this initial value $\tilde{{\btheta}}$, we choose $\tilde{\delta}$ such that $\tilde{\mbox{FPR}}_{n_0}(\tilde{{\btheta}}, \tilde{\delta}) = t$. In addition, we have also found that when $\tilde{\mbox{FPR}}_{n_0}$ is bounded by $t$, the performance of the optimization routine can be poor. Thus, we introduce another tuning parameter, $\alpha$, which allows for a small amount of relaxation in the constraint on the smooth approximation to the empirical FPR, imposing instead $\tilde{\mbox{FPR}}_{n_0}({\btheta},\delta) \leq t + \alpha.$ Since the effective sample size for the smooth approximation to the empirical FPR is $n_0$, we chose to scale $\alpha$ with $n_0$, and have found $\alpha = 1/(2n_0)$ to work well.

Our method involves computing the gradient of the smooth approximations to the true and false positive rates defined above, which is fast regardless of the number of biomarkers involved. This is in contrast with methods that rely on brute force (e.g., grid search), which typically become computationally infeasible for combinations of more than two or three biomarkers. However, we note that for any method, the risk of overfitting is expected to grow as the number of biomarkers increases relative to the sample size. We emphasize that our method does not impose constraints on the distribution of the biomarkers that can be included, except for weak conditions that allow us to establish its large-sample properties. 

\section{Simulations} 

\subsection{Bivariate Normal Biomarkers with Contamination}

First, we considered bivariate normal biomarkers with contamination, similar to a scenario described by~\cite{croux2003}. In particular, we considered a setting where two biomarkers $(X_1, X_2)$ were independently normally distributed with mean zero and variance one. $D$ was then defined as $D = 1(2X_1 + 2X_2 + \zeta > 0),$ where $\zeta$ was distributed as a logistic random variable with location parameter zero and scale parameter one. Next, the sample was contaminated by a set of points with $D = 0, X_1 = 6,$ and $X_2 = 6$. We consider simulations where the training set consisted of 800 or 1600 ``typical'' observations and 50 or 100, respectively, contaminating observations (this yielded a disease prevalence of approximately 47\%). The test set consisted of $10^6$ ``typical'' observations and 62,500 contaminating observations. The maximum acceptable FPR, $t$, was 0.2 or 0.3. We performed 1000 simulations. 

We considered five approaches: (1) logistic regression, (2) the robust logistic regression method proposed by~\cite{bianco1996}, (3) grid search, (4) the method proposed by~\cite{su1993} and (5) the proposed method. As discussed above, the method proposed by~\cite{su1993} yields a combination with maximum AUC when the biomarkers have a conditionally multivariate normal distribution. We did not consider the optimal AUC method proposed by~\cite{huang2011} as the implementation provided in \texttt{R} is too slow for use in simulations (and, as illustrated in Figure~\ref{fig1}, may not yield a combination with optimal TPR). While the methods recently proposed by~\cite{yan2018} to optimize the partial AUC are compelling and may yield a combination with high TPR at the specified FPR value, implementation of their method, particularly the nonparametric kernel-based method, is non-trivial, and so is not included here. Finally, the method of~\cite{liu2005b}, discussed above, may also yield a combination with high TPR at a particular FPR. However, given the shortcomings of this method described above (namely, that the range of FPRs over which the combination is optimal cannot be fixed in advanced and the biomarkers are assumed to have a conditionally multivariate normal distribution), we do not include this as a comparison method. Above all, none of these methods specifically target the TPR for a specified FPR which, as indicated by Figure~\ref{fig1}, may lead to combinations with reduced TPR at the specified FPR. 

We focused on evaluating the operating characteristics of the fitted combination rather than the biomarker coefficients as the former is typically of primary interest. In particular, we evaluated the TPR in the test data for FPR = $t$ in the test data. In other words, for each combination, the threshold used to calculate the TPR in the test data was chosen such that the FPR in the test data was equal to $t$. Evaluating the TPR in this way puts the combinations on equal footing in terms of the FPR, and so allows a fair comparison of the TPR. We evaluated the FPR of the fitted combinations in the test data using the thresholds estimated in the training data, i.e., the $(1-t)$th quantile of the fitted biomarker combination among controls in the training data. While we could have used the estimate of $\delta_t$ provided by our method in the evaluation, we found improved performance (that is, better control of the FPR) when re-estimating the threshold based on the fitted combination in the training data.

Table~\ref{contam} summarizes the results. For both sample sizes and FPR thresholds, all methods adequately controlled the FPR, while for the TPR, the proposed method outperformed logistic regression, robust logistic regression, and the method of~\cite{su1993}. Furthermore, the results from the proposed method were comparable to those from the grid search, which may be regarded as a performance reference but is infeasible for more than two or three biomarkers. 

\setlength{\tabcolsep}{3pt}

\begin{table}
\def~{\hphantom{0}}
\small
\caption{Mean TPR and FPR and corresponding standard deviation (in parentheses) in the test data across 1000 simulations for contaminated data with two biomarkers. The TPR is based on the threshold corresponding to an FPR of $t$ in the test data whereas the FPR is based on the thresholds estimated in the training data. $n$, size of the training dataset; $t$, acceptable FPR; TPR, true positive rate; FPR, false positive rate; GLM, standard logistic regression; rGLM, robust logistic regression; sTPR, proposed method. All numbers are percentages. }{
\begin{center}
\begin{tabular}{cc@{\hskip 0.2in}ccccc}
\\ \hline
$n$ & Measure & \multicolumn{5}{c}{Method} \\ \cline{3-7} \\ [-8pt]
& & GLM & rGLM & Grid Search & Su \& Liu & sTPR \\ \hline 
  \multicolumn{7}{c}{$t=$ 0.20} \\
800 & TPR & 52.7 (15.9) & 55.8 (15.1) & 72.5 (0.9) & 53.6 (15.6) & 72.0 (4.5) \\
  & FPR & 20.4 (1.4) & 20.4 (1.4) & 20.6 (1.3) & 20.4 (1.4) & 20.4 (1.3) \\ [1.5pt]
   1600 & TPR & 57.5 (13.2) & 60.0 (12.1) & 72.7 (0.5) & 58.3 (12.8) & 72.8 (0.5) \\
   & FPR & 20.3 (1.0) & 20.3 (1.0) & 20.4 (1.0) & 20.3 (1.0) & 20.3 (1.0) \\ \hline
  \multicolumn{7}{c}{$t=$ 0.30} \\
  800 & TPR & 68.5 (14.3) & 71.4 (13.6) & 86.0 (0.5) & 69.3 (13.9) & 86.0 (1.2) \\
  & FPR & 30.6 (1.8) & 30.6 (1.8) & 30.6 (1.8) & 30.6 (1.8) & 30.4 (1.8) \\ [1.5pt]
   1600 & TPR & 74.0 (11.5) & 76.1 (10.3) & 86.1 (0.3) & 74.7 (11.1) & 86.1 (0.3) \\
   & FPR & 30.4 (1.3) & 30.3 (1.3) & 30.4 (1.3) & 30.4 (1.3) & 30.2 (1.3) \\ \hline
\end{tabular}
\end{center}}
\label{contam}
\end{table}

\subsection{Conditionally Multivariate Lognormal Biomarkers}

We also considered simulations with conditionally multivariate lognormal biomarkers~\citep{mishrathesis}. In particular, we considered three biomarkers $(X_1,X_2,X_3)$. Among controls, $\mbox{log}(X_1)$ and $\mbox{log}(X_2)$ had a bivariate normal distribution with $E(\mbox{log}(X_1)) = 1.1, E(\mbox{log}(X_2)) = 1.1, Var(\mbox{log}(X_1)) = 0.04, Var(\mbox{log}(X_2)) = 0.5$ and $Cov(\mbox{log}(X_1),\mbox{log}(X_2)) = 0.09$. Among cases, $\mbox{log}(X_1)$ and $\mbox{log}(X_2)$ had a bivariate normal distribution with $E(\mbox{log}(X_1)) = 1, E(\mbox{log}(X_2)) = 1, Var(\mbox{log}(X_1)) = 0.05, Var(\mbox{log}(X_2)) = 0.05$ and $Cov(\mbox{log}(X_1),\mbox{log}(X_2)) = 0.015$. The third biomarker, $X_3$ was simulated from an independent lognormal distribution with $E(\mbox{log}(X_3)) = 1.65$ and $Var(\mbox{log}(X_3)) = 4.66$ among both cases and controls. Given the performance of the method of~\cite{su1993} and the performance of the proposed method relative to grid search observed above (and the computational challenges of implementing grid search for three biomarkers), we considered three methods here: (1) logistic regression, (2) robust logistic regression, and (3) the proposed method. Although neither logistic regression nor robust logistic regression performed particularly well in the simulations in Section 4.1, these methods represent the most commonly used approach for constructing biomarker combinations and the method used to provide starting values for the proposed method, respectively. Thus, it was important to include them here. 

The maximum acceptable FPR, $t$, was 0.2 and 1000 simulations were performed. The training data consisted of either 400 cases and 400 controls, or 800 cases and 800 controls. The test data consisted of $10^6$ observations. The TPR and FPR were evaluated as described above. We present the results in Table~\ref{anuidea}. All three methods did well in controlling the FPR at the specified value. Furthermore, the proposed method substantially outperformed logistic regression and robust logistic regression: the mean TPR based on the proposed method was at least 20\% larger than the mean TPRs from logistic regression and robust logistic regression. 

\begin{table}
\def~{\hphantom{0}}
\small
\caption{Mean TPR and FPR and corresponding standard deviation (in parentheses) in the test data across 1000 simulations for three conditionally lognormal biomarkers and $t=0.30$. The TPR is based on the threshold corresponding to an FPR of $t$ in the test data whereas the FPR is based on the thresholds estimated in the training data. $n$, size of the training dataset; $t$, acceptable FPR; TPR, true positive rate; FPR, false positive rate; GLM, standard logistic regression; rGLM, robust logistic regression; sTPR, proposed method. All numbers are percentages. }{
\begin{center}
\begin{tabular}{cc@{\hskip 0.2in}ccc}
\\ \hline
$n$ & Measure & \multicolumn{3}{c}{Method} \\ \cline{3-5} \\ [-8pt]
& & GLM & rGLM & sTPR \\ \hline 
800 & TPR & 34.1 (6.0) & 34.1 (6.0) & 41.5 (5.7) \\
  & FPR & 30.3 (2.3) & 30.3 (2.3) & 31.2 (2.4) \\ [1.5pt]
   1600 & TPR & 34.7 (4.2) & 34.7 (4.2) & 41.9 (4.9) \\
   & FPR & 30.2 (1.6) & 30.2 (1.6) & 30.7 (1.7) \\ \hline
\end{tabular}
\end{center}}
\label{anuidea}
\end{table}

\subsection{Bivariate Normal Biomarkers and Bivariate Normal Mixture Biomarkers}

The above simulations demonstrate superiority of our approach relative to alternative methods in particular scenarios. We conducted further simulations to demonstrate the feasibility of our approach in other settings (for instance, small sample size, small and large prevalence, and low FPR cutoffs) relative to logistic regression and robust logistic regression.  

We considered simulations with and without outliers in the data-generating distribution, and simulated data under a model similar to that used by~\cite{fong2016}. We considered two biomarkers $X_1$ and $X_2$ constructed as 
\begin{equation*}
\left( \begin{array}{c}
X_1 \\
X_2 \end{array} \right) = (1 - \Delta)\times Z_0 + \Delta \times Z_1,
\end{equation*}
where $\Delta$ was a Bernoulli random variable with success probability $\pi = 0.05$ when outliers were simulated and $\pi = 0$ otherwise, and $Z_0$ and $Z_1$ were independent bivariate normal random variables with mean zero and respective covariance matrices
\begin{equation*}
0.2\times \left(\begin{array}{cc}
1 & 0.9 \\
0.9 & 1 \end{array} \right), \,\, 2 \times \left(\begin{array}{cc}
1 & 0 \\
0 & 1 \end{array} \right).
\end{equation*}
$D$ was then simulated as a Bernoulli random variable with success probability 
$f\left\lbrace \beta_0 + 4 X_1 - 3 X_2 \right.$ $\left.- 0.8(X_1 - X_2)^3 \right\rbrace$. We considered two $f$ functions: $f_1(v) = \mbox{expit}(v) = e^v/(1+e^v)$ and a piecewise logistic function,
\begin{equation*}
f_2(v) = 1(v < 0) \times \frac{1}{1+e^{-v/3}} + 1(v \geq 0) \times \frac{1}{1+e^{-3v}}. 
\end{equation*} 
We varied $\beta_0$ to reflect varying prevalences, with a prevalence of approximately 50--60\% for $\beta_0 =$ 0, 16--18\% for $\beta_0 <$ 0, and 77--82\% for $\beta_0 >$ 0. We considered $t =$ 0.05, 0.1, and 0.2. A plot illustrating the data-generating distribution with $f = f_1$ and $\beta_0 =$ 0, with and without outliers, is given in Appendix D. 

The training data consisted of 200, 400, or 800 observations while the test set included $10^6$ observations. The TPR and FPR were evaluated as described above. The results are presented in Appendix C. When no outliers were present, the proposed method was comparable to logistic regression and robust logistic regression in terms of both the TPR and FPR. In the presence of outliers, robust logistic regression tended to provide combinations with higher TPRs than did logistic regression, and the TPRs of the combinations provided by the proposed method tended to be comparable to or somewhat better than the results from robust logistic regression. In all scenarios, all three methods controlled the FPR, particularly as sample size increased. In addition to demonstrating feasibility of our approach, these simulations highlight the fact that logistic regression is relatively robust to violations of the linear-logistic model (e.g., nonlinear biomarker combinations and deviations from the logit link). 

\subsection{Convergence}

In most simulation settings, convergence of the proposed method was achieved in more than $96\%$ of simulations. For some of the more extreme outlier scenarios considered in Section 4.3, convergence failed in up to $7.3\%$ of simulations. 

\section{Application to Diabetes Data}

We applied the method we have developed to a study of diabetes in women with Pima Indian heritage~\citep{smithpima}. We considered seven predictors measured in this study: number of pregnancies, plasma glucose concentration, diastolic blood pressure, triceps skin fold thickness, body mass index, age, and diabetes pedigree function (a measure of family history of diabetes~\citep{smithpima}). We used 332 observations as training data and reserved the remaining 200 observations for testing. The training and test datasets had 109 and 68 diabetes cases, respectively. We scaled the variables to have equal variance. The distribution of predictors is depicted in Appendix E. The combinations were fitted using the training data and evaluated using the test data. We fixed the acceptable FPR at $t = 0.10$. We used logistic regression, robust logistic regression, and the proposed method to construct the combinations, giving the results in Table \ref{PIDrslts}, where the fitted combinations from logistic regression and robust logistic regression have been normalized to aid in comparison. 

\begin{table}
\def~{\hphantom{0}}
\small
\caption{Fitted combinations of the scaled predictors in the diabetes study. GLM, standard logistic regression; rGLM, robust logistic regression; sTPR, proposed method with $t=$ 0.10.}{
\begin{center}
\begin{tabular}{l@{\hskip 0.35in}r@{\hskip 0.27in}r@{\hskip 0.27in}r}
\\ \hline
Predictor & GLM & rGLM & sTPR  \\ \hline 
Number of pregnancies & 0.321 & 0.320 & 0.403\\
Plasma glucose & 0.793 & 0.792 & 0.627 \\
Blood pressure & $-$0.077  & $-$0.073 & $-$0.026 \\
Skin fold thickness & 0.089 & 0.090 & $-$0.146\\
Body mass index & 0.399 & 0.400 & 0.609 \\
Diabetes pedigree & 0.280 & 0.281 & 0.191 \\
Age & 0.133 & 0.134 & 0.123  \\ \hline
\end{tabular}
\end{center}}
\label{PIDrslts}
\end{table}

Using thresholds based on an FPR of 10\% in the test data, the estimated TPR in the test data was 54.4\% for both logistic regression and robust logistic regression, and 55.9\% for the proposed method. The estimated FPR in the test data using thresholds corresponding to an FPR of 10\% in the training data was 18.2\% for both logistic regression and robust logistic regression and 26.5\% for the proposed method. The fact that these FPRs exceeded the target value for all three methods indicates potentially important differences in the controls between the training and test data.

\section{Discussion}

We have proposed a distribution-free method for constructing linear biomarker combinations by maximizing a smooth approximation to the TPR while constraining a smooth approximation to the FPR. Ours is the first distribution-free approach targeting the TPR for a specified FPR that can be used with more than two or three biomarkers. While we do not expect our method to outperform every other approach in every dataset, we have demonstrated broad feasibility of our method and, importantly, we have identified scenarios where the performance of our method is superior to alternative approaches. 


The proposed method could be adapted to minimize the FPR while controlling the TPR to be above some acceptable level. Since the TPR and FPR condition on disease status, the proposed method can be used with case-control data. For case-control data matched on a covariate, however, it becomes necessary to consider the covariate-adjusted ROC curve and corresponding covariate-adjusted summaries, and thus the methods presented here are not immediately applicable~\citep{janes2008}.

As our smooth approximation function is non-convex, the choice of starting values should be considered further. Extensions of convex methods, such as the ramp function method proposed by~\cite{fong2016} for the AUC, could also be considered. The idea of partitioning the search space, proposed by~\cite{yan2018}, may also be useful. Further research could investigate methods for evaluating the true and false positive rates of biomarker combinations after estimation, for example, sample-splitting, bootstrapping, or $k$-fold cross-validation.

\section{Software}

An \texttt{R} package containing code to implement the proposed method, \texttt{maxTPR}, is publicly available via CRAN. 

\section*{Funding} 

This work was supported by the National Institutes of Health [F31 DK108356, R01 CA152089, and R01 HL085757]; and the University of Washington Department of Biostatistics Career Development Fund [to M.C.]. The opinions, results, and conclusions reported in this article are those of the authors and are independent of the funding sources.

\clearpage

\section*{Appendix A}

\begin{prop}
\label{prop1}
If the biomarkers $(X_1,X_2)$ are conditionally multivariate normal with proportional covariance matrices given $D$, that is, 
\begin{equation*}
(X_1,X_2\mid D=0) \sim N({\mu}_0, \Sigma),\;\; (X_1,X_2\mid D=1) \sim N({\mu}_1, \sigma^2\Sigma),
\end{equation*}
then the optimal biomarker combination in the sense of the ROC curve is of the form
\begin{equation*}
\beta_0 + \beta_1 X_1 + \beta_2 X_2 + \beta_3 X_1X_2 + \beta_4 X_1^2 + \beta_5 X_2^2
\end{equation*}
for some vector $(\beta_0, \beta_1, \beta_2, \beta_3, \beta_4, \beta_5) \in \mathbb{R}^5$.
\end{prop}

\begin{proof}
It is known that the optimal combination of $(X_1,X_2)$ in terms of the ROC curve is the likelihood ratio, $f(X_1,X_2\mid D=1)/f(X_1,X_2\mid D=0)$, or any monotone increasing function thereof~\citep{mcintosh2002}. Let $\textbf{M} = (X_1,X_2)$. Without loss of generality, let ${\mu}_0 = 0$ and ${\mu}_1 = {\mu} = (\mu_{X_1}, \mu_{X_2})$. Then
\begin{align*}
\frac{f(M\mid D=1)}{f(M\mid D=0)} & = \frac{|\sigma^2 \Sigma|^{-1/2}\mbox{exp}\left\lbrace -\frac{1}{2}(M - {\mu})^{\top}(\sigma^2\Sigma)^{-1}(M - {\mu}) \right\rbrace}{|\Sigma|^{-1/2}\mbox{exp}\left\lbrace -\frac{1}{2}M^{\top}\Sigma^{-1}M \right\rbrace} \\
& = \frac{\mbox{exp}\left\lbrace -\frac{1}{2}(M - {\mu})^{\top}(\sigma^2\Sigma)^{-1}(M - {\mu}) \right\rbrace}{\sigma^2 \mbox{exp}\left\lbrace -\frac{1}{2}M^{\top}\Sigma^{-1}M \right\rbrace} \\
& = \frac{1}{\sigma^2} \mbox{exp}\left\lbrace -\frac{(M - {\mu})^{\top}\Sigma^{-1}(M - {\mu})}{2\sigma^2} + \frac{M^{\top}\Sigma^{-1}M}{2}\right\rbrace.
\end{align*}

Denote the entries of $\Sigma^{-1}$ by 
\begin{equation*}
\Sigma^{-1} = \left( \begin{array}{cc}
S_{11} & S_{12} \\
S_{21} & S_{22}
\end{array} \right).
\end{equation*}
Then, we can write that
\begin{equation*}
\begin{aligned}
&-\frac{1}{2\sigma^2}(M - {\mu})^{\top}\Sigma^{-1}(M - {\mu}) + \frac{1}{2}M^{\top}\Sigma^{-1}M \\
&\hspace{.5in}=\ \frac{1}{2}\left[\frac{1}{\sigma^2} \left\lbrace -S_{11}(X_1^2 - 2X_1\mu_{X_1} + \mu_{X_1}^2) - S_{21}(X_1X_2 - X_2\mu_{X_1} - X_1\mu_{X_2} + \mu_{X_1}\mu_{X_2}) \right. \right. \\
&\hspace{.5in}\hspace{.3in} \left. \left. - S_{12}(X_1X_2 - X_1\mu_{X_2} - X_2\mu_{X_1} + \mu_{X_1}\mu_{X_2}) - S_{22}(X_2^2 - 2X_2\mu_{X_2} + \mu_{X_2}^2)\right\rbrace \right. \\
&\hspace{.5in}\hspace{.3in} + S_{11}X_1^2 + S_{21}X_1X_2 + S_{12}X_1X_2 + S_{22}X_2^2 \bigg] \\
&\hspace{.5in}=\ \frac{1}{2}\left\lbrace \left( S_{11} - \frac{S_{11}}{\sigma^2} \right) X_1^2 + \left( S_{22} - \frac{S_{22}}{\sigma^2} \right) X_2^2 + \left( S_{12} + S_{21} - \frac{S_{12}}{\sigma^2} - \frac{S_{21}}{\sigma^2} \right) X_1X_2 \right. \\
&\hspace{.5in}\hspace{.3in} \left. + \left( \frac{2S_{11}\mu_{X_1} + S_{21}\mu_{X_2} + S_{12}\mu_{X_2}}{\sigma^2} \right)X_1 + \left( \frac{S_{21}\mu_{X_1} + S_{12}\mu_{X_1} + 2S_{22}\mu_{X_2}}{\sigma^2} \right)X_2 \right. \\
&\hspace{.5in}\hspace{.3in} \left. + \frac{-S_{11}\mu_{X_1}^2 - S_{21}\mu_{X_1}\mu_{X_2} - S_{12}\mu_{X_1}\mu_{X_2} - S_{22}\mu_{X_2}^2}{\sigma^2}\right\rbrace \\
&\hspace{.5in}=\ \beta_0 + \beta_1 X_1 + \beta_2 X_2 + \beta_3 X_1X_2 + \beta_4 X_1^2 + \beta_5 X_2^2
\end{aligned}
\end{equation*}
as claimed, where
\begin{align*}
\beta_0 &\ =\ \frac{-S_{11}\mu_{X_1}^2 - S_{21}\mu_{X_1}\mu_{X_2} - S_{12}\mu_{X_1}\mu_{X_2} - S_{22}\mu_{X_2}^2}{\sigma^2} \\
\beta_1 &\ =\ \left( \frac{2S_{11}\mu_{X_1} + S_{21}\mu_{X_2} + S_{12}\mu_{X_2}}{\sigma^2} \right) \\
\beta_2 &\ =\ \left( \frac{S_{21}\mu_{X_1} + S_{12}\mu_{X_1} + 2S_{22}\mu_{X_2}}{\sigma^2} \right) \\
\beta_3 &\ =\ \left( S_{12} + S_{21} - \frac{S_{12}}{\sigma^2} - \frac{S_{21}}{\sigma^2} \right) \\
\beta_4 &\ =\ \left( S_{11} - \frac{S_{11}}{\sigma^2} \right) \\
\beta_5 &\ =\ \left( S_{22} - \frac{S_{22}}{\sigma^2} \right).
\end{align*}

\end{proof}

\clearpage

\section*{Appendix B}

The proof of Theorem 1 relies on Lemmas 1 and 2, which are stated and proved below.

\begin{lem}
Say that a bounded function $f: \mathbb{R}^d \rightarrow \mathbb{R}$ and possibly random sets $\Omega_0,\Omega_1,\Omega_2,\ldots \subseteq \mathbb{R}^d$ are given, and let $\{a_n\}_{n\geq 1}$ be a decreasing sequence of positive real numbers tending to zero. For each $n\geq 1$, suppose that $\omega_{0,n}\in\Omega_0$ and $\omega_{n}\in\Omega_n$ are near-maximizers of $f$ over $\Omega_0$ and $\Omega_n$, respectively, in the sense that $f(\omega_{0,n})\geq \sup_{\omega\in\Omega_0}f(\omega)-a_n$ and $f(\omega_{n})\geq \sup_{\omega\in\Omega_n}f(\omega)-a_n$.
Further, define
\[ d_n = \adjustlimits \sup_{\omega \in \Omega_n} \inf_{\tilde{\omega} \in \Omega_0} d(\omega, \tilde{\omega}),\;\; e_n = \adjustlimits \sup_{\omega \in \Omega_0} \inf_{\tilde{\omega} \in \Omega_n} d(\omega, \tilde{\omega}), \] 
where $d$ is the Euclidean distance in $\mathbb{R}^d$. If $d_n$ and $e_n$ tend to zero almost surely, and $f$ is globally Lipschitz continuous, then $|f(\omega_{0,n}) - f(\omega_n)|$ tends to zero almost surely. In particular, this implies that \[\bigg|\sup_{\omega\in\Omega_0}f(\omega)-\sup_{\omega\in\Omega_n}f(\omega)\bigg|\longrightarrow 0\] almost surely.
\end{lem}

\begin{proof} 

Say that both $d_n$ and $e_n$ tend to zero almost surely, and denote by $K>0$ the Lipschitz constant of $f$. Suppose that for some $\epsilon > 0$ we have that \[P\bigg\{\limsup_n|f(\omega_{n})-f(\omega_{0,n})|>\epsilon\bigg\}>0\ .\] We will show that this leads to a contradiction, and thus that it must be true that 

\noindent $P\{\limsup_n|f(\omega_{n})-f(\omega_{0,n})|>\epsilon\}=0$ for each $\epsilon>0$, thus establishing the desired result.

On a set of probability one, there exists an $n_{\epsilon}\geq 1$ such that, for each $n\geq n_\epsilon$, there exists $\omega^*_{n}\in\Omega_0$ and $\omega^*_{0,n}\in\Omega_n$ satisfying $d(\omega^*_n,\omega_n)<\epsilon/(2K)$ and $d(\omega^*_{0,n},\omega_{0,n})<\epsilon/(2K)$. Then, on this same set, for $n \geq n_{\epsilon}$, $|f(\omega^*_{n}) - f(\omega_n)| \leq \epsilon/2$ and $|f(\omega^*_{0,n}) - f(\omega_{0,n})| \leq \epsilon/2$, so that $f(\omega_{0,n}) \leq f(\omega^*_{0,n}) + \epsilon/2$ and $f(\omega_n) \leq f(\omega^*_{n}) + \epsilon/2$ in particular. Since $\omega^*_{0,n}\in\Omega_n$ and $\omega_n^*\in\Omega_0$, it must also be true that $f(\omega^*_{0,n})\leq f(\omega_n)+a_n$ and $f(\omega^*_n)\leq f(\omega_{0,n})+a_n$. This then implies that $|f(\omega_{0,n}) - f(\omega_n)| \leq \epsilon/2+a_n$ for all $n \geq n_{\epsilon}$ on a set of probability one. Since $a_n$ tends to zero deterministically, this yields the sought contradiction.

To establish the last portion of the Lemma, we simply use the first part along with the fact that \[\bigg|\sup_{\omega\in\Omega_n}f(\omega)-\sup_{\omega\in\Omega_0}f(\omega)\bigg|\ \leq\ |f(\omega_{0,n})-f(\omega_n)| + 2a_n\ .\]
\end{proof}

\begin{lem}
Under conditions (1)--(5), we have that
\[ \sup_{({\btheta}, \delta) \in \Omega}\left| \tilde{\mbox{FPR}}_{n_0}({\btheta},\delta) - \mbox{FPR}({\btheta},\delta) \right| \longrightarrow 0 ,\,\, \sup_{({\btheta}, \delta) \in \Omega}\left| \tilde{\mbox{TPR}}_{n_1}({\btheta},\delta) - \mbox{TPR}({\btheta},\delta) \right| \longrightarrow 0 \]
almost surely as $n$ tends to $+\infty$, where $\Omega = \lbrace {(\btheta},\delta)\in  \mathbb{R}^p\times \mathbb{R}: ||{\btheta}||=1 \rbrace$.
\end{lem}

\begin{proof}
We prove the claim for the false positive rate (FPR); the proof for the true positive rate (TPR) is analogous. We can write
\begin{equation*}
\begin{aligned}
\sup_{({\btheta},\delta) \in \Omega} \left| \tilde{\mbox{FPR}}_{n_0}({\btheta},\delta) - \mbox{FPR}({\btheta},\delta) \right|\ &\leq\ \sup_{({\btheta},\delta) \in \Omega} \left| \tilde{\mbox{FPR}}_{n_0}({\btheta},\delta) - E\lbrace\tilde{\mbox{FPR}}_{n_0}({\btheta},\delta)\rbrace \right| \\
&\hspace{.3in}+ \sup_{({\btheta},\delta) \in \Omega} \left| E\lbrace\tilde{\mbox{FPR}}_{n_0}({\btheta},\delta)\rbrace - \mbox{FPR}({\btheta},\delta) \right|.
\end{aligned}
\end{equation*}
First, we consider $ \tilde{\mbox{FPR}}_{n_0}({\btheta},\delta) - E\lbrace\tilde{\mbox{FPR}}_{n_0}({\btheta},\delta)\rbrace $. We can write this as 
\[ \tilde{\mbox{FPR}}_{n_0}({\btheta},\delta) - E\lbrace\tilde{\mbox{FPR}}_{n_0}({\btheta},\delta)\rbrace = \frac{1}{n_0}\sum_{j=1}^{n_0} \Phi\left( \frac{{\btheta}^{\top}\textbf{X}_{0j} - \delta}{h} \right) - \int \Phi\left( \frac{{\btheta}^{\top}\textbf{x} - \delta}{h} \right) dF_0(\textbf{x})\ .\]
The class of functions $\mathcal{G}_1 = \lbrace ({\btheta}, \delta) \mapsto {\btheta}^{\top}\textbf{x} - \delta : {\btheta} \in \mathbb{R}^p, \delta \in \mathbb{R}, \textbf{x} \in \mathbb{R}^p \rbrace$ is a Vapnik--Chervonenkis (VC) class. Since $u\mapsto \Phi(u/h)$ is monotone for each $h>0$, the class of functions $\mathcal{G}_2 = \lbrace ({\btheta}, \delta) \mapsto \Phi\{({\btheta}^{\top}\textbf{x} - \delta)/h\} : {\btheta} \in \mathbb{R}^p, \delta \in \mathbb{R}, \textbf{x} \in \mathbb{R}^p, h > 0 \rbrace$ is also VC \citep{kosorok2008,vandervaartbook,vdvwellnerbook}. Since the constant 1 is an applicable envelope function for this class, $\mathcal{G}_2$ is $F_0$--Glivenko-Cantelli, giving that \citep{kosorok2008,vdvwellnerbook}
\[ \sup_{({\btheta},\delta) \in \Omega} \left| \tilde{\mbox{FPR}}_{n_0}({\btheta},\delta) - E\lbrace\tilde{\mbox{FPR}}_{n_0}({\btheta},\delta)\rbrace \right| \longrightarrow 0\] almost surely.

Next, we consider $ E\lbrace\tilde{\mbox{FPR}}_{n_0}({\btheta},\delta)\rbrace - \mbox{FPR}({\btheta},\delta) $. We can write this as
\begin{equation*}
E\lbrace\tilde{\mbox{FPR}}_{n_0}({\btheta},\delta)\rbrace - \mbox{FPR}({\btheta},\delta) = \int \Phi\left( \frac{{\btheta}^{\top}\textbf{x} - \delta}{h} \right) dF_0(\textbf{x}) - P({\btheta}^{\top}\textbf{X} > \delta \mid  D=0).
\end{equation*} 
For a general random variable $V$ with distribution function $F$ that is Lipschitz continuous, say with constant $M>0$, we can write
\begin{equation*}
E\left\lbrace \Phi\left( \frac{s-V}{h} \right) \right\rbrace = \int \Phi \left( \frac{s-v}{h} \right) dF(v) = h \int \Phi(u)f(s-hu)du
\end{equation*}
with $u = (s-v)/h$. Using integration by parts and Lemma 2.1 from~\cite{winter1979}, this becomes
\begin{equation*}
h \int \Phi(u)f(s-hu)du  = \int \phi(u)F(s-hu)du\ ,
\end{equation*}
and so, we find that 
\begin{equation*}
\begin{aligned}
\left| E\left\lbrace \Phi\left( \frac{s-V}{h} \right) \right\rbrace - F(s) \right|\ &=\ \left| \int \phi(u)F(s-hu)du - F(s) \right|\\
&\leq\ \int \left| F(s-hu) - F(s)\right| \phi(u) du 
\\
& \leq\ M\int |hu|\phi(u)du = M h \left(\frac{2}{\pi} \right)^{1/2}.
\end{aligned}
\end{equation*}
Since $h$ tends to zero as $n$ tends to infinity, this implies that
\[ \sup_s \left| E\left\lbrace \Phi\left( \frac{s-V}{h} \right) \right\rbrace - F(s) \right| = o(1)\ . \] 

We now return to ${\btheta}^{\top}\textbf{X}$ and consider the case $p=2$, so that ${\btheta}^{\top}\textbf{X} = \theta_1X_1 + \theta_2X_2.$ Let $Y_1 = \theta_1 X_1 + \theta_2 X_2$ and $Y_2 = \theta_2X_2$. Then, we have that $f_{Y_1,Y_2}(y_1,y_2) = f_{X_1,X_2}(x_1,x_2)|\theta_1\theta_2|^{-1}$, where $x_1=x_1(y_1,y_2)=(y_1-y_2)/\theta_1$ and $x_2=x_2(y_1,y_2)=y_2/\theta_2$. We find that
\begin{equation*}
\int \Phi\left( \frac{s-{\btheta}^{\top}\textbf{x}}{h} \right) dF_{\textbf{X}}(\textbf{x}) = \int \Phi\left( \frac{s-y_1}{h} \right) dF_{Y}(y) = \int\Phi\left( \frac{s-y_1}{h} \right) dF_{Y_1}(y_1)
\end{equation*}
for any $s \in \mathbb{R}$. Since $P({\btheta}^{\top}\textbf{X} \leq \delta \mid  D=0) = P(Y_1 \leq \delta \mid  D=0)$, we can write
\begin{align*}
&\sup_{({\btheta},\delta) \in \Omega} \left| \int \Phi\left( \frac{{\btheta}^{\top}\textbf{x} - \delta}{h} \right) dF_0(\textbf{x}) - P({\btheta}^{\top}\textbf{X} > \delta \mid  D=0) \right| \\
&\hspace{.2in}=\ \sup_{\delta \in \mathbb{R}} \left| \int \Phi\left( \frac{y_1 - \delta}{h} \right) dF_{Y_1|D=0}(y_1) - P(Y_1 > \delta \mid  D=0) \right| \\
&\hspace{.2in}=\ \sup_{\delta \in \mathbb{R}} \left| \int \Phi\left( \frac{\delta - y_1}{h} \right) dF_{Y_1|D=0}(y_1) - P(Y_1 \leq \delta \mid  D=0) \right|\ ,
\end{align*}
implying, in view of condition (4) and the results above, that
\begin{equation*}
\sup_{({\btheta},\delta) \in \Omega} \left| \int \Phi\left( \frac{{\btheta}^{\top}\textbf{x} - \delta}{h} \right) dF_0(\textbf{x}) - P({\btheta}^{\top}\textbf{X} > \delta \mid  D=0) \right| = o(1)\ .  
\end{equation*}
The result for $p > 2$ can be proved analogously.
 
Combining these results, we conclude that 
$\sup_{({\btheta},\delta) \in \Omega} | \tilde{\mbox{FPR}}_{n_0}({\btheta},\delta) - \mbox{FPR}({\btheta},\delta)|$ tends to zero almost surely, as claimed.
\end{proof}

\begin{proof}[Proof of Theorem 1] 
First, we show that $\limsup_n\mbox{FPR}(\tilde{{\btheta}}_{t},\tilde{\delta}_{t}) \leq t$ almost surely. We can write
\begin{align*}
\mbox{FPR}(\tilde{{\btheta}}_{t},\tilde{\delta}_{t})\ &=\ \tilde{\mbox{FPR}}_{n_0}(\tilde{{\btheta}}_{t},\tilde{\delta}_{t}) +  \lbrace\mbox{FPR}(\tilde{{\btheta}}_{t},\tilde{\delta}_{t}) - \tilde{\mbox{FPR}}_{n_0}(\tilde{{\btheta}}_{t},\tilde{\delta}_{t}) \rbrace \\
&\hspace{-.3in} \leq\ \tilde{\mbox{FPR}}_{n_0}(\tilde{{\btheta}}_{t},\tilde{\delta}_{t}) + |\mbox{FPR}(\tilde{{\btheta}}_{t},\tilde{\delta}_{t}) - \tilde{\mbox{FPR}}_{n_0}(\tilde{{\btheta}}_{t},\tilde{\delta}_{t})| \\
&\hspace{-.3in} \leq\ \tilde{\mbox{FPR}}_{n_0}(\tilde{{\btheta}}_{t},\tilde{\delta}_{t}) + \sup_{({\btheta},\delta) \in \Omega} |\mbox{FPR}({\btheta},\delta) - \tilde{\mbox{FPR}}_{n_0}({\btheta},\delta) |\ \leq\  t + \sup_{({\btheta},\delta) \in \Omega} |\mbox{FPR}({\btheta},\delta) - \tilde{\mbox{FPR}}_{n_0}({\btheta},\delta) |\ .
\end{align*} As such, it follows that \[P\{\limsup_n\mbox{FPR}(\tilde{\btheta}_t,\tilde{\delta}_t)\leq t\}\ \geq\ P\{\limsup_n \sup_{(\btheta,\delta)\in \Omega}
|\mbox{FPR}(\btheta,\delta)-\tilde{\mbox{FPR}}_{n_0}(\btheta,\delta)|=0\}\ =\ 1\] in view of Lemma 2, thereby establishing the first part of the theorem.

Let $t\in(0,1)$ be fixed. We now establish that \[\bigg| \mbox{TPR}(\tilde{{\btheta}}_{t}, \tilde{\delta}_{t}) - \sup_{({\btheta}, \delta) \in \Omega_t} \mbox{TPR}({\btheta}, \delta) \bigg| \longrightarrow 0\] almost surely.
For convenience, denote $(\btheta,\delta)$ by $\omega$. Consider the function $f$ defined pointwise as $f(\omega) = \mbox{TPR}(\btheta, \delta)$, and set $\Omega_0 = \Omega_t$ and $\Omega_n = \tilde{\Omega}_{t,n_0}$ for each $n\geq 1$. We verify that the conditions of Lemma 1 hold for these particular choices. We have that $f(\omega) = \mbox{TPR}(\btheta, \delta)$ is a bounded function. We must show $d_{n_0}$ and $e_{n_0}$ tend to zero almost surely, where 
\[ d_{n_0} = \adjustlimits \sup_{\omega \in \tilde{\Omega}_{t,n_0}} \inf_{\tilde{\omega} \in \Omega_t} d(\omega, \tilde{\omega}),\;\;\;\; e_{n_0} = \adjustlimits \sup_{\omega \in \Omega_t} \inf_{\tilde{\omega} \in \tilde{\Omega}_{t,n_0}} d(\omega, \tilde{\omega}), \] 
and $d$ is the Euclidean distance in $\mathbb{R}^{p+1}$. We consider $d_{n_0}$ first. Denote by $G_{\btheta}$ the conditional distribution function of $\btheta^\top \textbf{X}$ given $D=0$. By assumption, the corresponding conditional quantile function, denoted by $G_{\btheta}^{-1}$, is uniformly Lipschitz continuous over $\{\btheta\in\mathbb{R}^p:\|\btheta\|=1\}$, say with constant $C>0$ independent of $\btheta$. Suppose that, for some $\kappa > 0$, $\sup_{\omega \in \tilde{\Omega}_{t,n_0}} |\tilde{\mbox{FPR}}_{n_0}(\omega) - \mbox{FPR}(\omega)| \leq \kappa$. Because it is true that
\begin{equation*}
\kappa\ \geq\ \sup_{\omega \in \tilde{\Omega}_{t,n_0}} |\tilde{\mbox{FPR}}_{n_0}(\omega) - \mbox{FPR}(\omega)|\ \geq\ \bigg| \sup_{\omega \in \tilde{\Omega}_{t,n_0}} \tilde{\mbox{FPR}}_{n_0}(\omega) - \sup_{\omega \in \tilde{\Omega}_{t,n_0}} \mbox{FPR}(\omega)\bigg|\ ,
\end{equation*}
then $\sup_{\omega \in \tilde{\Omega}_{t,n_0}} \mbox{FPR}(\omega) \leq \kappa + t$, giving $\tilde{\mbox{FPR}}_{n_0}(\omega) \leq t$ and $\mbox{FPR}(\omega) \leq \kappa + t$ for each $\omega \in \tilde{\Omega}_{t,n_0}$.

For any given $\omega = ({\btheta}, \delta) \in \tilde{\Omega}_{t,n_0}$, write $t_*(\omega)=G_{\btheta}(\delta)$, giving $t_*(\omega)=\mbox{FPR}(\omega) \leq \kappa + t$. If $t_*(\omega)\leq t$, note also that $\omega\in \Omega_t$ and set $\omega_*=\omega$. Otherwise, find $\delta_*$ such that $1-G_{\btheta}(\delta_*) = t$, namely by taking $\delta_* = G_{\btheta}^{-1}(1-t)$. Defining $\omega_* = ({\btheta}, \delta_*) \in \Omega_t$, observe that
\[ d(\omega, \omega_*)\ =\  |\delta - \delta_*|\ =\ |G_{\btheta}^{-1}(1-t_*(\omega)) - G_{\btheta}^{-1}(1-t)|\ \leq\ C|t - t_*(\omega)|\ \leq\ C\kappa\ .\] Thus, for each $\omega \in \tilde{\Omega}_{t,n_0}$, it is true that $ \inf_{\tilde{\omega} \in \Omega_t} d(\omega, \tilde{\omega}) \leq C\kappa$ and therefore $d_{n_0} \leq C\kappa$. As such, if $d_{n_0} > \epsilon$ for some $\epsilon > 0$, then $\sup_{\omega \in \tilde{\Omega}_{t,n_0}} |\tilde{\mbox{FPR}}_{n_0}(\omega) - \mbox{FPR}(\omega)| > \kappa_{\epsilon}$ for $\kappa_\epsilon = \epsilon/C$.
This implies that
\[ P\bigg(\sup_{m \geq n_0} d_m > \epsilon \bigg)\ \leq\ P\bigg(\sup_{m \geq n_0} \sup_{\omega \in \tilde{\Omega}_{t,m}} |\tilde{\mbox{FPR}}_m(\omega) - \mbox{FPR}(\omega)| > \kappa_{\epsilon}\bigg)\ \longrightarrow\ 0 \]
by Lemma 2. Thus, $d_n$ tends to zero almost surely since, for each $\epsilon>0$, \[P\bigg(\limsup_{m}\{d_{m}\geq \epsilon\}\bigg)\ \leq\ P\bigg(\limsup_{m}d_{m}\geq \epsilon\bigg)\ =\ 0\ .\] Using similar arguments, we may show that $e_n$ also tends to zero almost surely.

The fact that $d_n$ and $e_n$ tend to zero almost surely implies, in view of Lemma 1, that we have that $|\sup_{({\btheta}, \delta) \in \Omega_t} \mbox{TPR}({\btheta}, \delta) - \sup_{({\btheta}, \delta) \in \tilde{\Omega}_{t,n_0}} \mbox{TPR}({\btheta}, \delta)|$ tends to zero almost surely. 
Combining this with an application of Lemma 2, we have that
\begin{align*}
&\bigg| \sup_{({\btheta}, \delta) \in \Omega_t} \mbox{TPR}({\btheta}, \delta) \, - \sup_{({\btheta}, \delta) \in \tilde{\Omega}_{t,n_0}} \tilde{\mbox{TPR}}_{n_1}({\btheta}, \delta) \bigg| \\
&\hspace{.3in}\leq\ \bigg| \sup_{({\btheta}, \delta) \in \Omega_t} \mbox{TPR}({\btheta}, \delta) - \sup_{({\btheta}, \delta) \in \tilde{\Omega}_{t,n_0}}\mbox{TPR}({\btheta}, \delta) \bigg| + \bigg| \sup_{({\btheta}, \delta) \in \tilde{\Omega}_{t,n_0}}\mbox{TPR}({\btheta}, \delta) - \sup_{({\btheta}, \delta) \in \tilde{\Omega}_{t,n_0}} \tilde{\mbox{TPR}}_{n_1}({\btheta}, \delta) \bigg| \\
&\hspace{.3in}\leq\ \bigg| \sup_{({\btheta}, \delta) \in \Omega_t} \mbox{TPR}({\btheta}, \delta) - \sup_{({\btheta}, \delta) \in \tilde{\Omega}_{t,n_0}}\mbox{TPR}({\btheta}, \delta) \bigg| + \sup_{({\btheta}, \delta) \in \tilde{\Omega}_{t,n_0}} | \mbox{TPR}({\btheta}, \delta) - \tilde{\mbox{TPR}}_{n_1}({\btheta}, \delta) | \ \longrightarrow\ 0 
\end{align*}
almost surely. Since $| \mbox{TPR}(\tilde{\btheta}_t, \tilde{\delta}_t) - \tilde{\mbox{TPR}}_{n_1}(\tilde{\btheta}_t, \tilde{\delta}_t) | \leq \sup_{({\btheta}, \delta) \in \Omega} | \mbox{TPR}({\btheta}, \delta) - \tilde{\mbox{TPR}}_{n_1}({\btheta}, \delta) |$ and, by Lemma 2, $\sup_{({\btheta}, \delta) \in \Omega} | \mbox{TPR}({\btheta}, \delta) - \tilde{\mbox{TPR}}_{n_1}({\btheta}, \delta) |$ tends to zero almost surely, $| \mbox{TPR}(\tilde{\btheta}_t, \tilde{\delta}_t) - \tilde{\mbox{TPR}}_{n_1}(\tilde{\btheta}_t, \tilde{\delta}_t) |$ tends to zero almost surely. In addition, since $(\tilde{\btheta}_t, \tilde{\delta}_t)$ is a near-maximizer of $\tilde{\mbox{TPR}}_{n_1}$, $\sup_{({\btheta}, \delta) \in \tilde{\Omega}_{t,n_0}} \tilde{\mbox{TPR}}_{n_1}({\btheta}, \delta) \leq \tilde{\mbox{TPR}}_{n_1}(\tilde{\btheta}_t, \tilde{\delta}_t) + a_n$, giving 
\begin{equation*}
\begin{aligned}
&\bigg| \sup_{({\btheta}, \delta) \in \Omega_t} \mbox{TPR}({\btheta}, \delta) - \mbox{TPR}(\tilde{\btheta}_t, \tilde{\delta}_t) \bigg| \\
&\hspace{.3in}\leq\ \bigg| \sup_{({\btheta}, \delta) \in \Omega_t} \mbox{TPR}({\btheta}, \delta) - \sup_{({\btheta}, \delta) \in \tilde{\Omega}_{t,n_0}} \tilde{\mbox{TPR}}_{n_1}({\btheta}, \delta) \bigg| + \bigg|\sup_{({\btheta}, \delta) \in \tilde{\Omega}_{t,n_0}} \tilde{\mbox{TPR}}_{n_1}({\btheta}, \delta) - \mbox{TPR}(\tilde{\btheta}_t, \tilde{\delta}_t) \bigg| \\
&\hspace{.3in}\leq\ \bigg| \sup_{({\btheta}, \delta) \in \Omega_t} \mbox{TPR}({\btheta}, \delta) - \sup_{({\btheta}, \delta) \in \tilde{\Omega}_{t,n_0}} \tilde{\mbox{TPR}}_{n_1}({\btheta}, \delta) \bigg| + \bigg|\sup_{({\btheta}, \delta) \in \tilde{\Omega}_{t,n_0}} \tilde{\mbox{TPR}}_{n_1}({\btheta}, \delta) - \tilde{\mbox{TPR}}_{n_1}(\tilde{\btheta}_t, \tilde{\delta}_t) \bigg| \\
&\hspace{.5in} + \bigg|\tilde{\mbox{TPR}}_{n_1}(\tilde{\btheta}_t, \tilde{\delta}_t) - \mbox{TPR}(\tilde{\btheta}_t, \tilde{\delta}_t) \bigg| \\
&\hspace{.3in}\leq\ \bigg| \sup_{({\btheta}, \delta) \in \Omega_t} \mbox{TPR}({\btheta}, \delta) - \sup_{({\btheta}, \delta) \in \tilde{\Omega}_{t,n_0}} \tilde{\mbox{TPR}}_{n_1}({\btheta}, \delta) \bigg| + a_n + \bigg|\tilde{\mbox{TPR}}_{n_1}(\tilde{\btheta}_t, \tilde{\delta}_t) - \mbox{TPR}(\tilde{\btheta}_t, \tilde{\delta}_t) \bigg|\ \longrightarrow\ 0
\end{aligned}
\end{equation*}
almost surely, completing the proof.
\end{proof}

\clearpage

\section*{Appendix C: Simulation results for data with outliers}

Note: When simulating with outliers, the true biomarker combination was occasionally so large that it returned a non-value for the outcome $D$; for example, with $f_1(v) = \mbox{expit}(v)$, this occurs in \texttt{R} when $v > 800$. These observations had to be removed from the simulated dataset, though this affected an extremely small fraction of observations. 

\begin{table}[ht!]
\def~{\hphantom{0}}
\small
\caption{Mean TPR and FPR and corresponding standard deviation (in parentheses) for $f(v) = f_1(v) \equiv \mbox{expit}(v) = e^v/(1+e^v)$ and $\beta_0$ = 0 across 1000 simulations. The TPR is based on the threshold corresponding to a FPR of $t$ in the test data whereas the FPR is based on the thresholds estimated in the training data. $n$, size of the training dataset; $t$, acceptable FPR; TPR, true positive rate; FPR, false positive rate; GLM, standard logistic regression; rGLM, robust logistic regression; sTPR, proposed method. All numbers are percentages. }{
\begin{center}
\begin{tabular}{ccc@{\hskip 0.2in}ccc}
\\ \hline
Outliers & $n$ & Measure & \multicolumn{3}{c}{Method} \\ \cline{4-6} \\ [-8pt]
& & & GLM & rGLM & sTPR \\ \hline 
\multicolumn{6}{c}{$t=$ 0.05} \\ [1pt]
Yes & 200 & TPR & 12.2 (2.1) & 13.6 (2.6) & 13.4 (2.7) \\
& & FPR & 5.7 (2.2) & 5.9 (2.3) & 6.4 (2.4) \\ [1.5pt]
   & 400 & TPR & 12.1 (1.7) & 14.1 (2.3) & 13.9 (2.4) \\
   & & FPR & 5.4 (1.6) & 5.4 (1.6) & 5.9 (1.7) \\ [1.5pt]
   & 800 & TPR & 11.8 (1.2) & 14.4 (2.2) & 14.4 (2.3) \\
   & & FPR & 5.1 (1.1) & 5.2 (1.1) & 5.5 (1.2) \\ [1.5pt]
  No & 200 & TPR & 18.3 (0.6) & 18.3 (0.6) & 17.8 (1.8) \\
  & & FPR & 5.5 (2.2) & 5.5 (2.2) & 6.2 (2.4) \\ [1.5pt]
   & 400 & TPR & 18.5 (0.3) & 18.5 (0.3) & 18.1 (1.6) \\
   & & FPR & 5.3 (1.5) & 5.3 (1.5) & 5.7 (1.6) \\ [1.5pt]
   & 800 & TPR & 18.6 (0.2) & 18.6 (0.2) & 18.4 (1.2) \\ 
   & & FPR & 5.2 (1.1) & 5.2 (1.1) & 5.5 (1.2) \\ \hline
  \multicolumn{6}{c}{$t=$ 0.10} \\ [1pt]
  Yes & 200 & TPR & 22.5 (3.8) & 24.6 (4.3) & 24.6 (4.2) \\
   & & FPR & 10.9 (3.1) & 11.1 (3.0) & 11.7 (3.2) \\ [1.5pt]
   & 400 & TPR & 21.8 (2.8) & 25.1 (4.0) & 25.2 (4.0) \\
   & & FPR & 10.4 (2.0) & 10.5 (2.1) & 11.0 (2.1) \\ [1.5pt]
   & 800 & TPR & 21.4 (2.0) & 25.7 (3.6) & 25.8 (3.6) \\
   & & FPR & 10.1 (1.5) & 10.1 (1.5) & 10.5 (1.5) \\ [1.5pt]
  No & 200 & TPR & 29.4 (0.8) & 29.5 (0.8) & 28.9 (2.2) \\
  & & FPR & 10.5 (3.1) & 10.5 (3.1) & 11.4 (3.2) \\ [1.5pt]
   & 400 & TPR & 29.8 (0.4) & 29.8 (0.4) & 29.5 (1.3) \\
   & & FPR & 10.4 (2.1) & 10.4 (2.1) & 10.9 (2.2) \\ [1.5pt]
   & 800 & TPR & 29.9 (0.2) & 29.9 (0.2) & 29.7 (1.5) \\
   & & FPR & 10.2 (1.5) & 10.2 (1.5)  & 10.6 (1.5) \\ \hline
  \multicolumn{6}{c}{$t=$ 0.20} \\ [1pt]
  Yes & 200 & TPR & 38.0 (5.1) & 40.8 (5.8) & 41.0 (5.7) \\
  & & FPR & 20.9 (4.0) & 21.1 (4.0) & 21.8 (4.0) \\ [1.5pt]
   & 400 & TPR & 37.4 (3.9) & 41.7 (5.3) & 41.9 (5.2) \\
   & & FPR & 20.5 (2.8) & 20.6 (2.9) & 21.1 (2.9) \\ [1.5pt]
   & 800 & TPR & 36.9 (2.9) & 42.4 (4.6) & 43.0 (4.4) \\
   & & FPR & 20.2 (2.0) & 20.4 (2.0) & 20.7 (2.0) \\ [1.5pt]
  No & 200 & TPR & 46.1 (0.9) & 46.1 (0.9) & 45.7 (1.5) \\
  & & FPR & 20.7 (4.1) & 20.8 (4.1) & 21.7 (4.2) \\  [1.5pt]
   & 400 & TPR & 46.4 (0.5) & 46.4 (0.5) & 46.2 (0.8) \\
   & & FPR & 20.3 (2.8) & 20.3 (2.8) & 21.0 (2.8) \\ [1.5pt]
   & 800 & TPR & 46.5 (0.2) & 46.5 (0.3) & 46.4 (0.6) \\
   & & FPR & 20.1 (2.0) & 20.1 (2.0) & 20.5 (2.0) \\ \hline
\end{tabular}
\end{center}}
\label{f1mod}
\end{table}

\begin{table}
\def~{\hphantom{0}}
\small
\caption{Mean TPR and FPR and corresponding standard deviation (in parentheses) for $f(v) = f_2(v) \equiv 1(v < 0)\times (1+e^{-v/3})^{-1} + 1(v \geq 0)\times (1+e^{-3v})^{-1}$ and $\beta_0$ = 0 across 1000 simulations. The TPR is based on the threshold corresponding to a FPR of $t$ in the test data whereas the FPR is based on the thresholds estimated in the training data. $n$, size of the training dataset; $t$, acceptable FPR; TPR, true positive rate; FPR, false positive rate; GLM, standard logistic regression; rGLM, robust logistic regression; sTPR, proposed method. All numbers are percentages. }{
\begin{center}
\begin{tabular}{ccc@{\hskip 0.2in}ccc}
\\ \hline
Outliers & $n$ & Measure & \multicolumn{3}{c}{Method} \\ \cline{4-6} \\ [-8pt]
& & & GLM & rGLM & sTPR \\ \hline 
\multicolumn{6}{c}{$t=$ 0.05} \\ [1pt]
Yes & 200 & TPR & 20.2 (7.3) & 26.4 (9.1) & 27.7 (9.2) \\
& & FPR & 5.9 (2.6) & 6.0 (2.6) & 6.5 (2.8) \\ [1.5pt]
   & 400 & TPR  & 19.0 (5.9) & 27.6 (8.5) & 29.3 (8.2) \\ 
   & & FPR & 5.5 (1.8) & 5.5 (1.7) & 5.8 (1.8) \\ [1.5pt]
   & 800 & TPR  & 17.9 (4.1) & 29.4 (7.5) & 30.8 (7.3) \\
   & & FPR & 5.3 (1.3) & 5.3 (1.2) & 5.5 (1.3) \\ [1.5pt]
  No & 200 & TPR  & 37.9 (1.7) & 37.8 (1.9) & 37.5 (3.1) \\
  & & FPR & 5.8 (2.7) & 5.7 (2.7) & 6.5 (2.9) \\ [1.5pt]
   & 400 & TPR  & 38.6 (0.9) & 38.5 (1.0) & 38.3 (2.1) \\
   & & FPR & 5.3 (1.8) & 5.3 (1.8) & 5.8 (1.8) \\ [1.5pt]
   & 800 & TPR  & 38.9 (0.4) & 38.9 (0.5) & 38.6 (2.2) \\
   & & FPR & 5.2 (1.3) & 5.2 (1.3) & 5.5 (1.3) \\ \hline
  \multicolumn{6}{c}{$t=$ 0.10} \\
  Yes & 200 & TPR  & 31.1 (8.9) & 37.4 (10.8) & 39.3 (11.0) \\
  & & FPR & 11.0 (3.5) & 11.3 (3.6) & 12.0 (3.6) \\ [1.5pt]
   & 400 & TPR  & 30.3 (7.1) & 39.9 (9.8) & 41.5 (9.6) \\
   & & FPR & 10.5 (2.5) & 10.7 (2.4) & 11.0 (2.5) \\ [1.5pt]
   & 800 & TPR  & 28.9 (5.0) & 41.1 (8.9) & 43.1 (8.6) \\
   & & FPR & 10.1 (1.7) & 10.3 (1.7) & 10.6 (1.8) \\ [1.5pt]
  No & 200 & TPR  & 48.2 (1.8) & 48.0 (1.9) & 48.2 (2.0) \\
  & & FPR & 10.9 (3.5) & 10.9 (3.5) & 11.7 (3.6) \\ [1.5pt]
   & 400 & TPR  & 48.8 (0.9) & 48.7 (1.0) & 48.7 (1.1) \\
   & & FPR & 10.4 (2.4) & 10.4 (2.4) & 10.9 (2.5) \\ [1.5pt]
   & 800 & TPR  & 49.2 (0.4) & 49.1 (0.5) & 49.0 (0.6) \\
   & & FPR & 10.2 (1.7) & 10.2 (1.7) & 10.7 (1.8) \\ \hline
  \multicolumn{6}{c}{$t=$ 0.20} \\
  Yes & 200 & TPR  & 45.0 (8.1) & 50.4 (9.8) & 51.9 (9.7) \\
  & & FPR & 21.2 (4.6) & 21.5 (4.7) & 22.0 (4.8) \\ [1.5pt]
   & 400 & TPR  & 44.4 (6.3) & 52.8 (8.6) & 54.0 (8.5) \\
   & & FPR & 20.4 (3.2) & 20.8 (3.3) & 21.2 (3.4) \\ [1.5pt]
   & 800 & TPR  & 44.1 (4.8) & 54.8 (7.3) & 56.5 (6.6) \\
   & & FPR & 20.2 (2.3) & 20.3 (2.3) & 20.7 (2.3) \\ [1.5pt]
  No & 200 & TPR  & 59.5 (1.3) & 59.4 (1.4) & 59.3 (1.8) \\
  & & FPR & 21.1 (4.6) & 21.1 (4.6) & 22.1 (4.7) \\ [1.5pt]
   & 400 & TPR  & 60.0 (0.6) & 59.9 (0.7) & 59.8 (0.9) \\
   & & FPR & 20.5 (3.4) & 20.6 (3.4) & 21.2 (3.4) \\ [1.5pt]
   & 800 & TPR  & 60.2 (0.4) & 60.1 (0.4) & 60.1 (0.5) \\
   & & FPR & 20.3 (2.2) & 20.3 (2.2) & 20.7 (2.3) \\ \hline
\end{tabular}
\end{center}}
\label{f2mod}
\end{table}

\begin{table}[ht!]
\def~{\hphantom{0}}
\small
\caption{Mean TPR and FPR and corresponding standard deviation (in parentheses) in the test data for $f(v) = f_1(v) \equiv \mbox{expit}(v) = e^v/(1+e^v)$ and $\beta_0$ = $-$1.75 across 1000 simulations. The TPR is based on the threshold corresponding to a FPR of $t$ in the test data whereas the FPR is based on the thresholds estimated in the training data. $n$, size of the training dataset; $t$, acceptable FPR; TPR, true positive rate; FPR, false positive rate; GLM, standard logistic regression; rGLM, robust logistic regression; sTPR, proposed method. All numbers are percentages. }{
\begin{center}
\begin{tabular}{ccc@{\hskip 0.2in}ccc}
\\ \hline
Outliers & $n$ & Measure & \multicolumn{3}{c}{Method} \\ \cline{4-6} \\ [-8pt]
& & & GLM & rGLM & sTPR \\ \hline 
\multicolumn{6}{c}{$t=$ 0.05} \\ [1pt]
  Yes & 200 & TPR & 13.0 (2.8) & 13.4 (3.4) & 13.5 (3.4) \\
  & & FPR & 5.3 (1.7) & 5.4 (1.7) & 5.7 (1.8) \\ [1.5pt]
   & 400 & TPR & 12.7 (1.9) & 13.4 (2.7) & 13.6 (2.9) \\
   & & FPR & 5.2 (1.2) & 5.2 (1.2) & 5.4 (1.2) \\ [1.5pt]
   & 800 & TPR & 12.5 (1.3) & 13.2 (2.1) & 13.6 (2.5) \\
   & & FPR & 5.1 (0.8) & 5.2 (0.8) & 5.2 (0.9) \\ [1.5pt]
  No & 200 & TPR & 18.1 (1.0) & 18.1 (1.1) & 17.5 (2.2) \\
  & & FPR & 5.5 (1.8) & 5.5 (1.8) & 5.9 (1.8) \\ [1.5pt]
   & 400 & TPR & 18.5 (0.6) & 18.5 (0.6) & 18.2 (1.6) \\
   & & FPR & 5.1 (1.2) & 5.2 (1.2) & 5.4 (1.3) \\ [1.5pt]
   & 800 & TPR & 18.7 (0.3) & 18.7 (0.3) & 18.5 (1.1) \\
   & & FPR & 5.1 (0.9) & 5.1 (0.9) & 5.3 (0.9) \\ \hline
  \multicolumn{6}{c}{$t=$ 0.10} \\ 
  Yes & 200 & TPR & 22.1 (4.5) & 22.7 (5.3) & 23.1 (5.3) \\
  & & FPR & 10.4 (2.4) & 10.5 (2.4) & 10.8 (2.4) \\ [1.5pt]
   & 400 & TPR & 21.9 (3.6) & 22.8 (4.7) & 23.4 (4.8) \\
   & & FPR & 10.1 (1.7) & 10.2 (1.7) & 10.4 (1.8) \\ [1.5pt]
   & 800 & TPR & 21.4 (2.3) & 22.3 (3.4) & 23.3 (4.3) \\
   & & FPR & 10.1 (1.2) & 10.1 (1.2) & 10.3 (1.2) \\ [1.5pt]
  No & 200 & TPR & 29.5 (1.3) & 29.4 (1.3) & 28.8 (2.5) \\
  & & FPR & 10.3 (2.3) & 10.4 (2.3) & 10.9 (2.3) \\ [1.5pt]
   & 400 & TPR & 29.8 (0.7) & 29.8 (0.7) & 29.5 (1.5) \\
   & & FPR & 10.2 (1.7) & 10.2 (1.7) & 10.6 (1.7) \\ [1.5pt]
   & 800 & TPR & 30.1 (0.4) & 30.1 (0.4) & 29.8 (1.1) \\
   & & FPR & 10.1 (1.1) & 10.1 (1.1) & 10.3 (1.1) \\ \hline
  \multicolumn{6}{c}{$t=$ 0.20} \\
  Yes & 200 & TPR & 36.4 (6.6) & 37.2 (7.8) & 38.1 (7.4) \\
  & & FPR & 20.5 (3.2) & 20.6 (3.1) & 21.0 (3.2) \\ [1.5pt]
   & 400 & TPR & 36.2 (4.7) & 37.3 (6.2) & 38.5 (6.4) \\
   & & FPR & 20.1 (2.2) & 20.2 (2.3) & 20.4 (2.2) \\ [1.5pt]
  & 800 & TPR & 35.7 (3.0) & 37.0 (4.6) & 38.8 (5.7) \\
  & & FPR & 20.2 (1.5) & 20.2 (1.5) & 20.4 (1.5) \\ [1.5pt]
  No & 200 & TPR & 46.1 (1.7) & 46.1 (1.7) & 45.5 (2.6) \\
  & & FPR & 20.4 (3.1) & 20.5 (3.2) & 21.0 (3.2) \\ [1.5pt]
  & 400 & TPR & 46.7 (0.8) & 46.7 (0.8) & 46.4 (1.3) \\
  & & FPR & 20.1 (2.1) & 20.2 (2.1) & 20.5 (2.1) \\ [1.5pt]
  & 800 & TPR & 47.0 (0.4) & 47.0 (0.4) & 46.8 (0.7) \\
  & & FPR & 20.0 (1.6) & 20.0 (1.6) & 20.2 (1.6) \\ \hline
\end{tabular}
\end{center}}
\label{f1low}
\end{table}

\begin{table}
\def~{\hphantom{0}}
\small
\caption{Mean TPR and FPR and corresponding standard deviation (in parentheses) in the test data for $f(v) = f_1(v) \equiv \mbox{expit}(v) = e^v/(1+e^v)$ and $\beta_0$ = 1.75 across 1000 simulations. The TPR is based on the threshold corresponding to a FPR of $t$ in the test data whereas the FPR is based on the thresholds estimated in the training data. $n$, size of the training dataset; $t$, acceptable FPR; TPR, true positive rate; FPR, false positive rate; GLM, standard logistic regression; rGLM, robust logistic regression; sTPR, proposed method. All numbers are percentages. }{
\begin{center}
\begin{tabular}{ccc@{\hskip 0.2in}ccc}
\\ \hline
Outliers & $n$ & Measure & \multicolumn{3}{c}{Method} \\ \cline{4-6} \\ [-8pt]
& & & GLM & rGLM & sTPR \\ \hline 
\multicolumn{6}{c}{$t=$ 0.05} \\ [1pt]
Yes & 200 & TPR & 8.4 (1.2) & 8.4 (1.4) & 8.2 (1.8) \\
& & FPR & 7.3 (4.0) & 6.9 (3.9) & 7.7 (4.6) \\ [1.5pt]
   & 400 & TPR & 8.6 (0.9) & 8.5 (1.1) & 8.3 (1.6) \\
   & & FPR & 6.3 (2.7) & 6.3 (2.7) & 6.7 (2.9) \\ [1.5pt]
   & 800 & TPR & 8.7 (0.6) & 8.6 (0.7) & 8.5 (1.5) \\
   & & FPR & 5.8 (1.8) & 5.8 (1.8) & 6.1 (2.0) \\ [1.5pt]
  No & 200 & TPR & 18.7 (1.0) & 18.7 (1.0) & 17.2 (3.5) \\
  & & FPR & 6.3 (4.1) & 6.1 (4.0) & 7.4 (4.5) \\ [1.5pt]
   & 400 & TPR & 19.0 (0.5) & 19.0 (0.6) & 17.9 (2.9) \\
   & & FPR & 5.7 (2.7) & 5.6 (2.7) & 6.4 (3.0) \\ [1.5pt]
   & 800 & TPR & 19.2 (0.3) & 19.2 (0.3) & 18.3 (2.9) \\
   & & FPR & 5.3 (1.9) & 5.3 (1.9) & 5.9 (2.0) \\ \hline
  \multicolumn{6}{c}{$t=$ 0.10} \\
  Yes & 200 & TPR & 18.6 (3.9) & 19.1 (4.7) & 19.4 (4.8) \\
  & & FPR & 12.4 (5.1) & 12.4 (5.0) & 13.5 (5.4) \\ [1.5pt]
   & 400 & TPR & 18.6 (2.5) & 19.2 (3.5) & 19.8 (3.8) \\
   & & FPR & 11.1 (3.4) & 11.1 (3.5) & 12.0 (3.7) \\ [1.5pt]
   & 800 & TPR & 18.4 (1.4) & 19.2 (2.6) & 19.8 (3.4) \\
   & & FPR & 10.8 (2.6) & 10.8 (2.6) & 11.3 (2.7) \\ [1.5pt]
  No & 200 & TPR & 29.9 (1.3) & 29.9 (1.3) & 28.7 (3.6) \\
  & & FPR & 11.7 (5.2) & 11.5 (5.2) & 13.1 (5.6) \\ [1.5pt]
   & 400 & TPR & 30.4 (0.6) & 30.3 (0.7) & 29.4 (3.4) \\
   & & FPR & 10.7 (3.6) & 10.6 (3.6) & 11.7 (3.8) \\ [1.5pt]
   & 800 & TPR & 30.6 (0.3) & 30.6 (0.3) & 30.2 (2.0) \\
   & & FPR & 10.4 (2.5) & 10.4 (2.5) & 11.1 (2.5) \\ \hline
  \multicolumn{6}{c}{$t=$ 0.20} \\
  Yes & 200 & TPR & 34.2 (6.4) & 34.9 (7.7) & 35.9 (7.1) \\
  & & FPR & 22.5 (6.5) & 22.7 (6.3) & 24.0 (6.7) \\ [1.5pt]
   & 400 & TPR & 34.2 (4.3) & 35.0 (5.6) & 36.3 (5.9) \\
   & & FPR & 21.4 (4.7) & 21.5 (4.7) & 22.4 (4.8) \\ [1.5pt]
   & 800 & TPR & 33.9 (2.8) & 35.0 (4.4) & 36.2 (5.0) \\
   & & FPR & 20.6 (3.3) & 20.7 (3.3) & 21.3 (3.4) \\ [1.5pt]
  No & 200 & TPR & 46.4 (1.6) & 46.4 (1.6) & 45.6 (3.3) \\
  & & FPR & 22.2 (7.0) & 22.0 (7.0) & 23.9 (7.1) \\ [1.5pt]
   & 400 & TPR & 47.0 (0.8) & 47.0 (0.8) & 46.5 (2.2) \\
   & & FPR & 20.8 (5.0) & 20.7 (4.9) & 22.0 (5.0) \\ [1.5pt]
   & 800 & TPR & 47.2 (0.4) & 47.2 (0.4) & 46.9 (1.9) \\
   & & FPR & 20.6 (3.4) & 20.6 (3.4) & 21.4 (3.5) \\ \hline
\end{tabular}
\end{center}}
\label{f1high}
\end{table}

\begin{table}
\def~{\hphantom{0}}
\small
\caption{Mean TPR and FPR and corresponding standard deviation (in parentheses) in the test data for $f(v) = f_2(v) \equiv 1(v < 0)\times (1/(1+e^{-v/3})) + 1(v \geq 0)\times (1/(1+e^{-3v}))$ and $\beta_0$ = $-$5.25 across 1000 simulations. The TPR is based on the threshold corresponding to a FPR of $t$ in the test data whereas the FPR is based on the thresholds estimated in the training data. $n$, size of the training dataset; $t$, acceptable FPR; TPR, true positive rate; FPR, false positive rate; GLM, standard logistic regression; rGLM, robust logistic regression; sTPR, proposed method. All numbers are percentages. }{
\begin{center}
\begin{tabular}{ccc@{\hskip 0.2in}ccc}
\\ \hline
Outliers & $n$ & Measure & \multicolumn{3}{c}{Method} \\ \cline{4-6} \\ [-8pt]
& & & GLM & rGLM & sTPR \\ \hline 
\multicolumn{6}{c}{$t=$ 0.05} \\ [1pt]
Yes & 200 & TPR & 7.1 (1.1) & 7.1 (1.1) & 7.1 (1.1) \\
& & FPR & 5.7 (1.8) & 5.7 (1.8) & 5.9 (1.9) \\ [1.5pt]
   & 400 & TPR & 7.4 (1.0) & 7.3 (0.9) & 7.3 (1.0) \\
   & & FPR & 5.3 (1.2) & 5.4 (1.2) & 5.5 (1.2) \\ [1.5pt]
   & 800 & TPR & 7.6 (0.8) & 7.5 (0.8) & 7.5 (0.9) \\
   & & FPR & 5.1 (0.8) & 5.2 (0.8) & 5.2 (0.9) \\ [1.5pt]
  No & 200 & TPR & 7.3 (1.4) & 7.3 (1.4) & 7.2 (1.4) \\
  & & FPR & 5.5 (1.7) & 5.6 (1.7) & 5.9 (1.8) \\ [1.5pt]
   & 400 & TPR & 7.8 (0.9) & 7.8 (1.0) & 7.7 (1.1) \\
   & & FPR & 5.2 (1.2) & 5.2 (1.1) & 5.4 (1.2) \\ [1.5pt]
   & 800 & TPR & 8.1 (0.4) & 8.1 (0.4) & 8.0 (0.7) \\
   & & FPR & 5.1 (0.9) & 5.1 (0.9) & 5.2 (0.9) \\ \hline
  \multicolumn{6}{c}{$t=$ 0.10} \\
  Yes & 200 & TPR & 12.4 (2.0) & 12.3 (2.0) & 12.4 (2.0) \\
  & & FPR & 10.6 (2.3) & 10.6 (2.3) & 10.9 (2.4) \\ [1.5pt]
   & 400 & TPR & 12.6 (1.7) & 12.4 (1.7) & 12.6 (1.8) \\
   & & FPR & 10.4 (1.7) & 10.4 (1.7) & 10.6 (1.7) \\ [1.5pt]
   & 800 & TPR & 12.8 (1.5) & 12.5 (1.5) & 12.7 (1.6) \\
   & & FPR & 10.2 (1.2) & 10.2 (1.1) & 10.3 (1.2) \\ [1.5pt]
  No & 200 & TPR & 13.9 (2.2) & 13.9 (2.2) & 13.6 (2.3) \\
  & & FPR & 10.7 (2.3) & 10.8 (2.3) & 11.2 (2.4) \\ [1.5pt]
   & 400 & TPR & 14.5 (1.5) & 14.5 (1.5) & 14.4 (1.6) \\
   & & FPR & 10.2 (1.6) & 10.2 (1.6) & 10.5 (1.6) \\ [1.5pt]
   & 800 & TPR & 15.0 (0.8) & 15.0 (0.8) & 14.9 (1.0) \\
   & & FPR & 10.2 (1.2) & 10.2 (1.2) & 10.4 (1.2) \\ \hline
  \multicolumn{6}{c}{$t=$ 0.20} \\
  Yes & 200 & TPR & 22.4 (3.6) & 22.2 (3.7) & 22.5 (3.7) \\
  & & FPR & 20.9 (3.1) & 20.9 (3.2) & 21.3 (3.2) \\ [1.5pt]
   & 400 & TPR & 22.6 (3.3) & 22.3 (3.3) & 22.7 (3.3) \\
   & & FPR & 20.6 (2.2) & 20.6 (2.2) & 20.9 (2.2) \\ [1.5pt]
   & 800 & TPR & 22.8 (2.7) & 22.3 (2.8) & 22.8 (2.8) \\
   & & FPR & 20.2 (1.5) & 20.2 (1.6) & 20.4 (1.6) \\ [1.5pt]
  No & 200 & TPR & 25.8 (3.5) & 25.7 (3.5) & 25.5 (3.6) \\
  & & FPR & 20.9 (3.1) & 20.9 (3.1) & 21.4 (3.1) \\ [1.5pt]
   & 400 & TPR & 26.9 (2.3) & 26.9 (2.3) & 26.8 (2.3) \\
   & & FPR & 20.5 (2.1) & 20.5 (2.1) & 20.8 (2.1) \\ [1.5pt]
   & 800 & TPR & 27.7 (1.1) & 27.7 (1.1) & 27.5 (1.3) \\
   & & FPR & 20.3 (1.6) & 20.3 (1.6) & 20.5 (1.6) \\ \hline
\end{tabular}
\end{center}}
\label{f2low}
\end{table}

\begin{table}
\def~{\hphantom{0}}
\small
\caption{Mean TPR and FPR and corresponding standard deviation (in parentheses) in the test data for $f(v) = f_2(v) \equiv 1(v < 0)\times (1/(1+e^{-v/3})) + 1(v \geq 0)\times (1/(1+e^{-3v}))$ and $\beta_0$ = 0.6 across 1000 simulations. The TPR is based on the threshold corresponding to a FPR of $t$ in the test data whereas the FPR is based on the thresholds estimated in the training data. $n$, size of the training dataset; $t$, acceptable FPR; TPR, true positive rate; FPR, false positive rate; GLM, standard logistic regression; rGLM, robust logistic regression; sTPR, proposed method. All numbers are percentages. }{
\begin{center}
\begin{tabular}{ccc@{\hskip 0.2in}ccc}
\\ \hline
Outliers & $n$ & Measure & \multicolumn{3}{c}{Method} \\ \cline{4-6} \\ [-8pt]
& & & GLM & rGLM & sTPR \\ \hline 
\multicolumn{6}{c}{$t=$ 0.05} \\ [1pt]
Yes & 200 & TPR & 23.0 (8.6) & 30.5 (10.9) & 31.9 (10.8) \\
& & FPR & 6.4 (3.3) & 6.3 (3.4) & 6.8 (3.7) \\ [1.5pt]
  Yes & 400 & TPR & 21.5 (6.9) & 31.8 (10.5) & 33.5 (10.1) \\
  & & FPR & 5.8 (2.3) & 5.8 (2.4) & 6.2 (2.6) \\ [1.5pt]
  Yes & 800 & TPR & 20.0 (4.4) & 34.6 (9.2) & 35.8 (8.5) \\
  & & FPR & 5.4 (1.6) & 5.3 (1.6) & 5.7 (1.7) \\ [1.5pt]
  No & 200 & TPR & 49.7 (1.5) & 49.5 (1.7) & 48.6 (4.4) \\
  & & FPR & 6.0 (3.5) & 5.9 (3.5) & 6.8 (3.7) \\ [1.5pt]
  No & 400 & TPR & 50.3 (0.7) & 50.1 (0.8) & 49.7 (2.3) \\
  & & FPR & 5.5 (2.5) & 5.4 (2.5) & 6.1 (2.6) \\ [1.5pt]
  No & 800 & TPR & 50.5 (0.4) & 50.5 (0.5) & 50.1 (2.0) \\
  & & FPR & 5.2 (1.6) & 5.2 (1.6) & 5.6 (1.7) \\ \hline
  \multicolumn{6}{c}{$t=$ 0.10} \\
  Yes & 200 & TPR & 37.3 (11.0) & 45.7 (13.7) & 48.4 (13.2) \\
  & & FPR & 11.5 (4.5) & 11.6 (4.5) & 12.4 (4.6) \\ [1.5pt]
  Yes & 400 & TPR & 35.2 (8.5) & 47.5 (12.9) & 50.6 (12.1) \\
  & & FPR & 10.8 (3.1) & 10.9 (3.2) & 11.4 (3.3) \\ [1.5pt]
  Yes & 800 & TPR & 34.5 (6.6) & 51.3 (10.7) & 53.6 (10.2) \\
  & & FPR & 10.4 (2.2) & 10.4 (2.2) & 10.8 (2.3) \\ [1.5pt]
  No & 200 & TPR & 61.3 (1.4) & 61.1 (1.6) & 60.7 (3.2) \\
  & & FPR & 10.9 (4.5) & 10.9 (4.5) & 12.1 (4.7) \\ [1.5pt]
  No & 400 & TPR & 61.8 (0.7) & 61.6 (0.8) & 61.4 (1.2) \\
  & & FPR & 10.6 (3.2) & 10.6 (3.2) & 11.4 (3.3) \\ [1.5pt]
  No & 800 & TPR & 62.0 (0.4) & 62.0 (0.4) & 61.8 (0.8) \\
  & & FPR & 10.3 (2.3) & 10.3 (2.3) & 10.9 (2.4) \\ \hline
  \multicolumn{6}{c}{$t=$ 0.20} \\
  Yes & 200 & TPR & 53.2 (10.6) & 60.9 (13.0) & 64.2 (12.3) \\
  & & FPR & 21.2 (5.9) & 21.8 (6.0) & 22.8 (6.0) \\ [1.5pt]
  Yes & 400 & TPR & 52.0 (8.5) & 63.5 (11.8) & 65.4 (11.3) \\
  & & FPR & 20.7 (4.1) & 21.1 (4.2) & 21.7 (4.1) \\ [1.5pt]
  Yes & 800 & TPR & 51.1 (6.0) & 66.3 (9.7) & 68.6 (8.2) \\
  & & FPR & 20.4 (3.0) & 20.6 (3.0) & 21.1 (3.0) \\ [1.5pt]
  No & 200 & TPR & 73.3 (1.1) & 73.1 (1.3) & 73.0 (1.5) \\
  & & FPR & 21.4 (6.4) & 21.4 (6.4) & 22.5 (6.3) \\ [1.5pt]
  No & 400 & TPR & 73.6 (0.6) & 73.5 (0.7) & 73.5 (0.8) \\
  & & FPR & 20.7 (4.4) & 20.7 (4.4) & 21.6 (4.4) \\ [1.5pt]
  No & 800 & TPR & 73.8 (0.3) & 73.8 (0.4) & 73.8 (0.4) \\
  & & FPR & 20.4 (3.0) & 20.4 (3.0) & 21.0 (3.0) \\ \hline 
\end{tabular}
\end{center}}
\label{f2high}
\end{table}

\clearpage

\section*{Appendix D: Illustration of data simulated with outliers}

\begin{figure}[ht!]
\begin{center}
\includegraphics[scale=0.43]{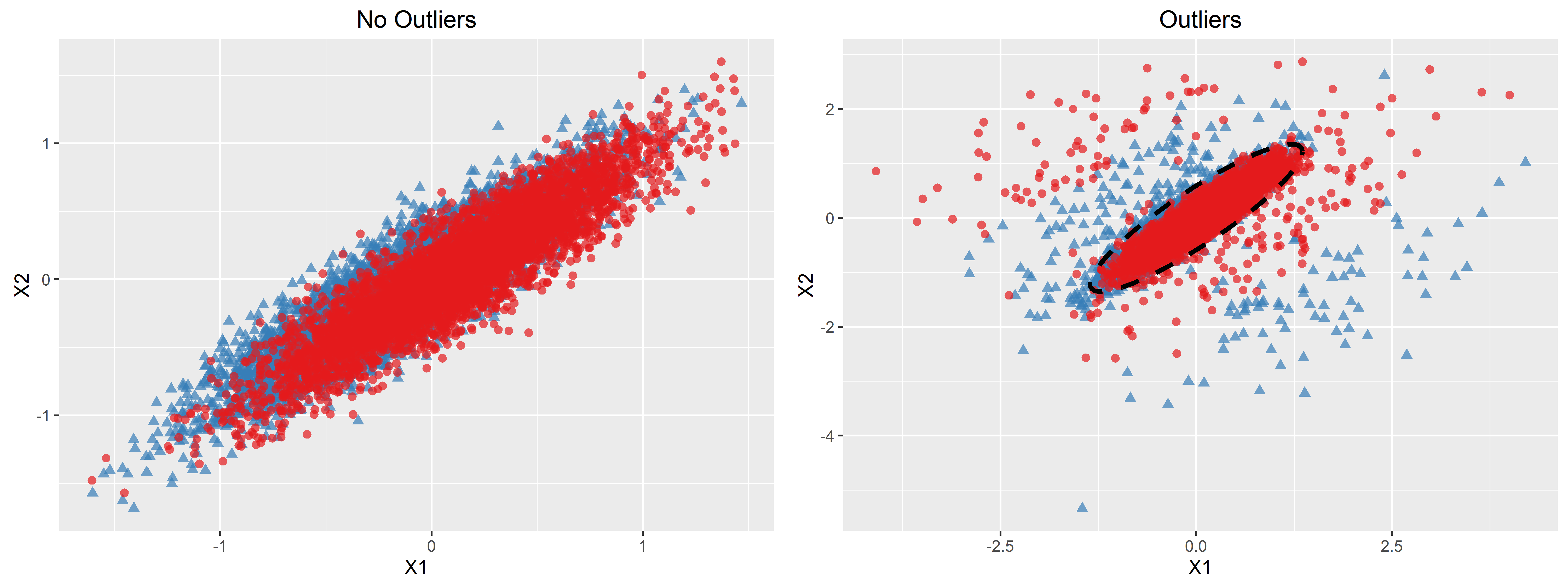}
\captionsetup{labelformat=empty}
\caption{\textbf{Figure:} Datasets with $f(v) = f_1(v) \equiv \mbox{expit}(v)$, $\beta_0 = 0$, without (left plot) and with outliers (right plot). Cases are represented by red circles, and controls are represented by blue triangles. The plot with outliers also includes an ellipse (dashed black line) indicating the 99\% confidence region for the distribution of $(X_1, X_2)$ without outliers.}
\label{OutlierPlot}
\end{center}
\end{figure}

\clearpage

\section*{Appendix E: Biomarker distribution in Pima Indian diabetes dataset}

\begin{figure}[ht!]
\begin{center}
\includegraphics[scale=0.7]{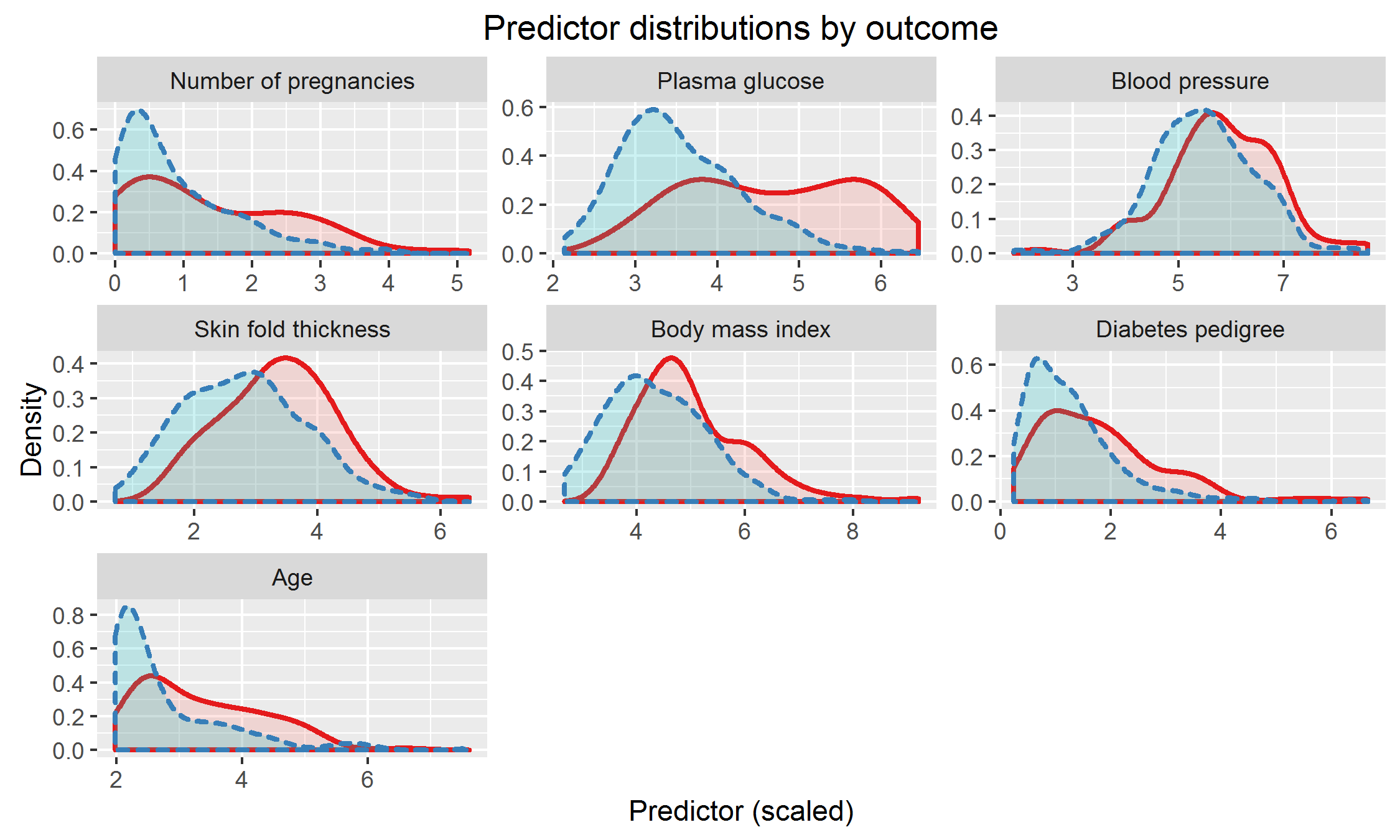}
\captionsetup{labelformat=empty}
\caption{\textbf{Figure:} Stratified distributions of the scaled predictors measured in the diabetes study for the observations in the training data. The predictors are number of pregnancies, plasma glucose concentration, diastolic blood pressure, triceps skin fold thickness, body mass index, diabetes pedigree function, and age. The predictor values are shown on the x-axis of each plot. The red solid line represents the distribution among diabetes cases and the blue dotted line represents the distribution among controls.}
\label{diabmarkers}
\end{center}
\end{figure}

\end{document}